\documentclass[11pt,a4paper]{article}

\newif\ifprivate
\privatetrue

\newif\ifarxiv
\arxivtrue

\usepackage{microtype}
\usepackage{a4wide}
\usepackage[T1]{fontenc}
\usepackage[utf8]{inputenc}
\usepackage{authblk}

\usepackage[table]{xcolor}
\definecolor{lilla}{HTML}{750787}
\usepackage{amssymb,amsmath,amsthm}
\usepackage{nicefrac}
\usepackage{hyperref}
\usepackage[ruled,vlined,linesnumbered]{algorithm2e}

\usepackage[capitalise]{cleveref}
\usepackage{graphicx}
\usepackage{caption}
\usepackage{subcaption}
\usepackage{booktabs}
\usepackage{tabularx}
\usepackage{multicol}
\usepackage{gensymb}
\usepackage{enumerate}   
\usepackage{csquotes}
\usepackage[hyperpageref]{backref}
\usepackage{paralist}

\ifprivate{}
	\usepackage[text size=footnotesize,color=green!15!white]{todonotes}
\else{}
	\usepackage[text size=footnotesize,color=green!15!white,disable]{todonotes}
\fi{}

\renewcommand*{\backref}[1]{}
\renewcommand*{\backrefalt}[4]{\ifcase #1\or [p.~#2]\else [pp.~#2]\fi }

\hypersetup{
	colorlinks,citecolor=green!50!black,linkcolor=red!50!black,urlcolor=black
}

\usepackage[numbers]{natbib}
\usepackage{tikz}

\usepackage{mathtools}
\usepackage{xargs}
\newcommandx{\set}[2][1=1]{\ensuremath{\{#1,\ldots,#2\}}}
\newcommandx{\tlog}[3][1=,3=]{\log_{#1}^{#3}(#2)}
\newcommandx{\ith}[2][1=th]{#2\nobreakdash-#1}

\newtheorem{theorem}{Theorem}
\newtheorem{lemma}[theorem]{Lemma}

\newtheorem{proposition}[theorem]{Proposition}
\newtheorem{observation}[theorem]{Observation}
\newtheorem{corollary}[theorem]{Corollary}
\newtheorem{claim}[theorem]{Claim}
\theoremstyle{definition}

\newtheorem{rrule}{Reduction rule}

\crefname{observation}{Observation}{Observations}
\crefname{rrule}{Reduction Rule}{Reduction Rules}
\crefname{arule}{Augmentation Rule}{Augmentation Rules}
\crefname{construction}{Construction}{Constructions}
\crefname{theorem}{Theorem}{Theorems}
\Crefname{theorem}{Thm.}{Thms.}
\crefname{corollary}{Corollary}{Corollaries}
\crefname{lemma}{Lemma}{Lemmata}
\Crefname{corollary}{Cor.}{Cors.}
\crefname{proposition}{Proposition}{Propositions}
\Crefname{proposition}{Prop.}{Props.}
\crefname{algorithm}{Algorithm}{Algorithms}

\newcommand{\calX}{\mathcal{X}}
\newcommand{\yes}{\textnormal{\texttt{yes}}}

\newcommand{\RD}{$(\Rightarrow)\:$}
\newcommand{\LD}{$(\Leftarrow)\:$}

\crefname{problem}{Problem}{Problems}
\Crefname{problem}{Prob.}{Probs.}

\newcommand{\boxproblem}[3]{
	\begin{center}   
		\fbox{~\begin{minipage}{.97\textwidth}
				\vspace{2pt} 
				\noindent
				\normalsize\textsc{#1}
				\vspace{1pt}
				
				\setlength{\tabcolsep}{3pt}
				\renewcommand{\arraystretch}{1.0}
				\begin{tabularx}{\textwidth}{@{}lX@{}}
					\normalsize\textbf{Input:}       & \normalsize#2 \\
					\normalsize\textbf{Question:}    & \normalsize#3 \\
				\end{tabularx}
		\end{minipage}}
	\end{center}
}

\usepackage{dsfont}
\newcommand{\N}{\mathbb{N}}

\newcommand{\bigO}{\mathcal{O}}

\newcommand{\cocl}[1]{\ensuremath{\operatorname{#1}}}
\newcommand{\W}[1]{\cocl{W[#1]}}

\newcommand{\NP}{\cocl{NP}}
\newcommand{\FPT}{\cocl{FPT}}
\newcommand{\fpt}{fixed-parameter tractable}

\newcommand{\XP}{\cocl{XP}}

\newcommand{\calC}{\mathcal{C}}

\newcommand{\prob}[1]{{\normalfont\textsc{#1}}}
\newcommand{\cec}{\prob{Colored Clustering}}
\newcommand{\chc}{\prob{Colored Hypergraph Clustering}}
\newcommand{\cecS}{\prob{CC}}

\newcommand{\colors}{\ensuremath{C}}
\newcommand{\coloring}{\ensuremath{\ell}}
\newcommand{\numcolors}{\ensuremath{\lvert C\rvert}}

\DeclareMathOperator{\lfe}{lfe}

\DeclareMathOperator{\secw}{secw}

\DeclarePairedDelimiter{\abs}{\lvert}{\rvert}

\newcommand{\wilog}{without loss of generality}
\newcommand{\Wilog}{Without loss of generality}

\newcommand{\tv}{\tilde{v}}

\newcommand{\tikzpreamble}{
	\def\nsc{0.5}
	\tikzstyle{xnodeRed}=[circle,scale=\nsc,draw,fill=red];
	\tikzstyle{xnodeBlue}=[circle,scale=\nsc,draw,fill=blue];
	\tikzstyle{xnodeGreen}=[circle,scale=\nsc,draw,fill=black!30!green];
	\tikzstyle{xnodeBlack}=[circle,scale=\nsc,draw,fill=black];
	\tikzstyle{xnode}=[circle,scale=\nsc,draw];
	\tikzstyle{xnodeSmall}=[circle,scale=\nsc*0.65,draw];
	\tikzstyle{satE} = [ultra thick];
	\tikzstyle{unsatE} = [thick,dotted];
	\definecolor{lblue}{RGB}{151, 197, 232}
	\definecolor{lred}{RGB}{245, 193, 193}
}

\newcommand{\thetitle}{Parameterized Algorithms for Colored Clustering}

\date{}

\title{\thetitle} 

\author[1]{Leon Kellerhals}

\author[1]{Tomohiro Koana\thanks{Supported by the DFG Project DiPa, NI 369/21.}}

\author[1,2]{Pascal Kunz\thanks{Supported by the DFG Research Training Group 2434 ``Facets of Complexity''.}}

\author[1]{Rolf Niedermeier\thanks{%
We dedicate this paper to Rolf, who tragically passed away recently.
We are deeply affected by this loss of our co-author, colleague, and advisor.
Rolf has contributed tremendously to computer science and, in particular, to multivariate algorithms, and should have continued doing so for a long time.
The computer science community will build on the foundations he has laid.
}}

\affil[1]{\small Technische Universit\"at Berlin, Algorithmics and Computational Complexity, Berlin, Germany\\ \texttt{\{leon.kellerhals,\,tomohiro.koana,\,p.kunz.1\}@tu-berlin.de}}

\affil[2]{\small Humboldt-Universit\"at zu Berlin, Algorithm Engineering, Berlin, Germany}

\begin{document}
\maketitle

\begin{abstract}
	In the \cec{} problem, one is asked to cluster edge-colored (hyper\nobreakdash-)graphs whose colors represent interaction types.
	More specifically, the goal is to select as many edges as possible without choosing two edges that 
	share an endpoint and are colored differently.
	Equivalently, the goal can also be described as assigning colors to the vertices in a way that fits the edge-coloring as well as possible.
	
	As this problem is \NP-hard, we build on previous work by studying its parameterized complexity.
	We give a $2^{\bigO(k)}\cdot n^{\bigO(1)}$-time algorithm where $k$ is the number of edges to be selected and $n$ the number of vertices.
	We also prove the existence of a problem kernel of size~$\bigO(k^{5/2})$, resolving an open problem posed in the literature.
	We consider parameters that are smaller than $k$, the number of edges to be selected, and $r$, the number of edges that can be deleted.
	Such smaller parameters are obtained by considering the difference between $k$ or $r$ and some lower bound on these values.
	We give both algorithms and lower bounds for \cec{} with such parameterizations.
	Finally, we settle the parameterized complexity of \cec{} with respect to structural graph parameters by showing that it is \W{1}-hard with respect to both vertex cover number and tree-cut width, but fixed-parameter tractable with respect to slim tree-cut width.
\end{abstract}

\section{Introduction}
\label{sect:intro}

Graph clustering is one of the most fundamental tasks in analyzing data that captures interactions between entities.
The idea is that if two vertices are in the same cluster, then the corresponding entities are similar in terms of their interactions.
Typical approaches lead to meaningful clusterings whenever the edges model interactions of the same type, possibly with weights.
Those approaches are not, however,  designed to deal with data that captures interactions of different types.
There are several settings in which the clusters should capture similarity not only in terms of interactions, but in terms of types of interactions.
For instance, brain coactivation graphs capture which brain regions are active or inactive at the same time when exposed to certain types of stimuli~\cite{Crossley2013}.
In the Drug Abuse Warning Network,\footnote{\url{https://www.samhsa.gov/data/data-we-collect/dawn-drug-abuse-warning-network}} interactions describe combinations of drugs taken by patients prior to an ER visit.
Other settings with similar interaction categorization by type include coauthorship networks (categorized by publication venue) and copurchasing networks (categorized by type of purchase).

\citet{Angel2016} introduced an approach to finding such category-sensitive clusters, which we will call \cec{} (\cecS{}):
Given an edge-colored graph, the goal is to color the vertices in a way that maximizes the number of \emph{stable} edges:
edges whose endpoints are both assigned the same color as the edge.
\citet{Angel2016} proved the problem to be \NP-hard and gave approximation algorithms as well as tractable special cases for the problem.
Their approximation algorithm was improved by \citet{Ageev2014} and \citet{Alhamdan2019}.
\citet{Cai2018} studied the parameterized complexity of the problem, giving FPT algorithms with respect to the number~$k$ of stable edges\footnote{Roughly speaking, an FPT algorithm with respect to $k$ has a running time of~$f(k) \cdot n^{\bigO(1)}$, where $k$ is the parameter and $n$ is the instance size.} and to the number~$r$ of unstable edges.
\citet{Amburg2020} introduced the problem under a different name and for hypergraphs.
\citet{Veldt2022} continued the study of the problem on hypergraphs.
In several of the examples listed above, it is sensible to model the interactions as hypergraphs.

\paragraph{Related work.}
Another clustering model which also captures edge colors is chromatic correlation clustering~\cite{Bonchi15}, in which the goal is to cluster the vertices such that the number of edge modifications (additions, deletions, recolorings) to obtain a disjoint set of monochromatic cliques is minimized.
This approach penalizes missing edges, that is, edges that need to be added such that each cluster becomes a clique.
Chromatic correlation clustering generalizes correlation clustering~\cite{Bansal04} and thus is NP-hard even for one color.
\cec{} does not penalize such missing edges and becomes more tractable: it is polynomial-time solvable for two colors.
Closely related to graph clustering in edge-colored graphs is the field of clustering multi-layer graphs~\cite{Mucha10,Chen2018}.
In that scenario, the layers do not relate to a cluster type.
Finally, hypergraph clustering has been studied extensively~\cite{Agarwal06,Papa07,Gleich18,Fukunaga19,Li20}.%

\paragraph{Our contributions.}
We continue the study of the parameterized complexity of \cec{} on graphs and on hypergraphs.
For the parameter $k$, the number of stable edges, we improve on the FPT result due to \citet{Cai2018} by giving a single-exponential time algorithm.
We also prove that \cec{} admits a polynomial kernel for the parameter, thus answering an open question by \citet{Cai2018}.
Both results translate to hypergraphs with constant-sized edges.

As the problem is FPT with respect to both the number of stable and unstable edges, we consider above guarantee parameterizations.
We introduce the concept in the corresponding section and show that it can be used to obtain fixed-parameter algorithms for parameters that are smaller than the number of stable edges and a parameter smaller than the number of unstable edges.
Again, both results can be lifted to work for hypergraphs.
We complement these results with hardness proofs for above-guarantee parameters that are even slightly smaller.
We also consider structural graph parameters.
The problem is hard even for fairly large structural parameters:
We show that there is presumably no FPT algorithm for the parameters vertex cover number and tree-cut width.
However, the problem is FPT when parameterized by the slim tree-cut width.

Finally, we close a gap in the classical tractability of \cec{} by proving that the problem is NP-hard on graphs even if every vertex has degree at most three and on hypergraphs even if every vertex has degree at most two.

\section{Preliminaries}
\label{sect:prelim}
\subsection{Graphs and problem definition}
For standard graph terminology, we refer to~\citet{Diestel2017}.
A \emph{hypergraph} $G=(V,E)$ consists of a finite vertex set $V$ and an edge set $E\subseteq 2^V$.
It is a \emph{graph} if $\abs{e} = 2$ for all $e \in E$.
The order of $G$ is $\max_{e\in E} \abs{e}$.
Let $v \in V$ be a vertex.
We denote its set of \emph{incident} edges by $\delta_G(v) \coloneqq \{ e \in E \mid v \in e \}$ and its \emph{neighborhood} by $N_G(v) \coloneqq (\bigcup_{e \in \delta(v)} e) \setminus \{v\}$.
The \emph{degree} of $v$ is $\deg_G(v) \coloneqq |\delta_G(v)|$.
For an edge set $F \subseteq E$, let $\deg_{G,F}(v) = |\delta(v) \cap F|$.
For a vertex set $X \subseteq V$, let $G[X]$ denote the subgraph induced  by $X$.
Let $\colors$ be a finite set of colors.
An \emph{edge coloring} is a function $\ell\colon E \to \colors$ and a \emph{vertex coloring} a function $f\colon V \to \colors$.
The \emph{chromatic degree} of $v \in V$ is $\deg^\chi_G(v) \coloneqq \{ c \in \colors \mid \exists e \in \delta(v) \colon \ell(e) =c\}$
For any $c\in \colors$, the color-$c$ degree of $v\in V$ is $\deg_{G,c} (v) \coloneqq |\{ e \in \delta(v) \mid \ell(e) = c\}|$.
We drop the subscript $\cdot_G$ whenever it is clear from the context.
A set of edges $F \subseteq E$ in a hypergraph $G=(V,E)$ with an edge coloring $\ell \colon E \to \colors$ is \emph{stable} if $\ell(e) = \ell(e')$ for all $e,e' \in F$ with $e\cap e' \neq \emptyset$.
In other words, all edges that share a connected component in the hypergraph $(V,F)$ must have the same color.
This leads to the computational problem of selecting a largest set of stable edges:

\boxproblem{\textsc{Colored (Hypergraph) Clustering (CC/CHC)}}
{A (hyper)-graph $G=(V,E)$, an edge coloring~$\coloring \colon E \to \colors$, and an integer $k \in \N$.}
{Does $G$ contain a stable edge set $F$ of size at least $k$?}
Sometimes it is more convenient to express stability using vertex colorings.
An edge $e\in E$ is \emph{stable} under the vertex coloring $f\colon V\to \colors$, if $f(v) = \ell(e)$ for all $v\in E$ and unstable otherwise.

One can transform any stable edge set $F$ into a vertex coloring $f_F \colon V \to \colors$ by setting:
\begin{align*}
	f_F(v) \coloneqq
	\begin{cases}
		\ell(e), & \text{if } v \in e \in F,\\
		\bot, & \text{if } v \text{ is not incident to any edges in } F,
	\end{cases}
\end{align*}
where $\bot \in \colors$ is an arbitrary default color.
We note the following:
\begin{observation}
	If $F\subseteq E$ is stable, then the edges in~$F$ are stable under~$f_F$.
\end{observation}

\subsection{Parameterized complexity}
Let~$\Sigma$ be a
finite alphabet.
A \emph{parameterized problem}~$L\subseteq \Sigma^*\times \mathbb N$ is a subset of all instances~$(x,\kappa)$ in~$\Sigma^*\times \mathbb N$,
and~$\kappa$ is the \emph{parameter}.
A parameterized problem~$L$ is 
\begin{inparaenum}[(i)]
	\item \emph{fixed-parameter tractable} (or contained in the class \FPT) if there is an algorithm that decides~$L$ in~$f(\kappa)\cdot |x|^{O(1)}$ time,  
	\item contained in the class \XP{} if there is an algorithm that decides~$L$ in~$|x|^{f(\kappa)}$ time, and
	\item \emph{para-NP-hard} if $L$ is~\NP-hard for any constant value of the parameter,
\end{inparaenum}
where~$f\colon \N \to \N$ is any computable function that only depends on the parameter.
Note that $\FPT\subseteq \XP$.
For running time bounds, we use the~$\bigO^*$ notation which hides factors that are polynomial in the input size.
If a parameterized problem is \emph{\W{1}-hard}, 
then it is presumably not in \FPT, and if it is para-NP-hard, then it is not in~\XP{} (unless P$\;=\;$\NP).
A \emph{kernel} for $L$ is a polynomial-time algorithm that takes the instance $(x,\kappa)$ and outputs a second instance $(x',\kappa')$ such that
\begin{inparaenum}[(i)]
	\item $(x,\kappa) \in L \iff (x',\kappa') \in L$ and
	\item $|(x',\kappa')| \leq f(\kappa)$ for a computable function $f$.
\end{inparaenum}
The \emph{size} of the kernel is $f$.
For further details, we refer to the standard literature~\cite{Cygan2015,Downey2013}.

\section{Parameterizing by the number of stable edges}

Cai and Leung~\cite{Cai2018} showed that \cec{} is \fpt{} with respect to the maximum number $k$ of stable edges.
Their algorithm uses color-coding and can find a coloring such that at least $k$ edges are stable with probability $1-\varepsilon$ in $\bigO(k^{2k}\ln(\frac{1}{\varepsilon})(n+m))$ time.
This algorithm can be derandomized to yield a deterministic algorithm with running time $k^{2k+\bigO(\log k)}(n+m)$.
In the following, we will improve on this running time and give a single-exponential algorithm for \cecS{} parameterized by $k$.
Cai and Leung asked if \cecS{} has a kernel that is polynomial in~$k$.
We will show that the problem does, indeed, admit a kernel of size~$\bigO(k^{5/2})$.

\subsection{Single-exponential time algorithm}

Our single-exponential time algorithm is a parameterized reduction to \textsc{Weighted Exact Cover}, defined as follows.
\boxproblem{\textsc{Weighted Exact Cover}}
{A universe~$U$, a family~$\mathcal S$ of nonempty subsets of~$U$, a weight function~$w \colon \mathcal S \to \mathbb N$, and integers~$s,W$.}
{Is there a subfamily~$\mathcal S'$ of pairwise disjoint subsets  with $\abs{\bigcup_{S\in \mathcal S'} S} =s$ such that $\sum_{S \in \mathcal S'} w(S) \ge W$?}

Using the fact that \textsc{Weighted Exact Cover} can be solved in~$\bigO(2.851^s |S|\cdot|U| \log^2 |U|)$ time~\cite{goyal2015deterministic,shachnai2017graphpartitioning}, we prove the following.

\begin{theorem}
	\label{thm:fpt-stable-edges}
	\cec{} can be solved in $\bigO^*(2^{\bigO(k)})$ time.
\end{theorem}
\begin{proof}
	We provide a parameterized reduction to \textsc{Weighted Exact Cover}.
	Let~$(G=(V,E), \coloring, k)$ be an instance of \cec{}.
	For each color~$c \in \colors$ denote by~$G^c$ the subgraph spanned by the edges of color~$c$ (note that $G^c$ does not contain isolated vertices).
	Let~$q_c$ be the number of connected components in~$G^c$ and let~$Q^c_{1}, Q^c_{2}, \dots, Q^c_{q_c} \subseteq V$ be the connected components.
	If $q_c \ge k$ for some color~$c$ or $G^c[Q^c_p]$ has at least $k$ edges for some~$c$ and~$p \in [q_c]$, then this is a trivial \yes-instance.
	So assume the contrary.
	We create an instance~$(U, \mathcal S, s\coloneqq 2k, W\coloneqq k)$ as follows.
	Set~$U \coloneqq V \uplus \{x_1,\ldots,x_s\}$.
	For each color~$c \in \colors$, for each component~$Q^c_p$, $p \in [q_c]$, and for each nonempty subset $X \subseteq Q^c_p$, we add~$X$ to~$\mathcal S$ and set~$w(X) \coloneqq \abs{\{ \{u,v\} \in E \mid u,v\in X, \ell(\{u,v\}) = c\}}$, that is, the number of edges of color~$c$ in the graph induced by~$X$.
	We also add the singletons $\{x_1\},\ldots,\{x_s\}$ to $\mathcal S$ and set $w(\{x_i\}) \coloneqq 0$.

	Suppose that the instance of \cecS{} is a \yes-instance, that is, there exists a stable edge set~$F \subseteq E$ of size at least~$k$.
	Then, for each connected component $X \subseteq V$ of $G_F = (V, F)$ of size at least two, there is a set $X \in \mathcal{S}$ whose weight is the number of edges in $G_F[X]$, since the edges of $G_F[X]$ have the same color.
	These connected components are pairwise disjoint.
	We create $\mathcal S'$ by adding all of these sets $X$ and, if $\abs{\bigcup_{S\in \mathcal S'} S} < s$, an appropriate number of singletons $\{x_i\}$.
	Thus, $\mathcal S'$ covers exactly $s$ elements and its weight is at least $W$.

	Conversely, suppose that there exists a subfamily~$\mathcal S'$ with weight at least~$W$ and $\abs{\bigcup_{S\in \mathcal S'} S} \allowbreak = s$.
	Each subset in $\mathcal S'$ except for the singletons corresponds to a set~$V'$ of vertices such that, by construction, $G[V']$ contains~$w(V')$ edges, all having the same color.
	Hence, setting~$f(v) = c$ for each~$v \in V'$ yields a solution for the \cecS{} instance with at least~$k$ stable edges.

	As for the running time, note that the instance of \textsc{Weighted Exact Cover} contains at most~$\numcolors \cdot k \cdot 2^k + 2k$ sets in~$\mathcal S$.
	Using the aforementioned algorithm yields a running time of~$\bigO(2.851^s \cdot (\numcolors \cdot k \cdot 2^k + 2k) (|V| + 2k)\log^2(\abs{V} + 2k)) = \bigO(16.26^k\cdot k^2 \log^2 k\cdot \abs{V}\log^2\abs{V})$.
\end{proof}

The above result can be generalized to hypergraphs of order~$d$.
Since a connected component of a hypergraph with at most~$k-1$ edges may have at most~$d(k-1)$ vertices,
the instance of \textsc{Weighted Exact Cover} has at most~$|C| \cdot k \cdot 2^{dk}+dk$ sets.
\begin{corollary}
	\chc{} can be solved in~$\bigO^*(2^{\bigO(d^2k)})$ time.
\end{corollary}

\subsection{A polynomial kernel}
We next show that \cec{} admits a kernel of size polynomial in the number~$k$ of stable edges, answering an open of question posed by Cai and Leung~\cite{Cai2018}.

We begin with the observation that a matching is stable:

\begin{rrule}
	\label{rr:matching}
	If $G$ has a matching of size $k$, then return \yes.
\end{rrule}

Let $M$ be any maximal matching and let $S := \bigcup_{e \in M} e$.
Then, $S$ is a vertex cover of size at most $2k$ by \Cref{rr:matching}.
Thus, it remains to bound the size of the independent set $V \setminus S$.
It can be arbitrarily large at this point;
in fact, we show in \Cref{sec:struct} that \cec{} is W[1]-hard when parameterized by the vertex cover number, and thus, it is presumably not possible to bound the size of $I$ in terms of the vertex cover size alone.

We use another simple observation that a set of monochromatic edges is stable:

\begin{rrule}
	\label{rr:one-color}
	If there are $k$ edges of the same color, then return \yes.
\end{rrule}

In view of \Cref{rr:one-color}, we only need to bound the number of colors.
We can achieve this by bounding the chromatic degree of every vertex.

\begin{rrule}
	\label{rr:redundant-color}
	If there is a vertex $v \in V$ of chromatic degree at least~$2k+1$,
	then delete all edges with color~$c$ in $\delta(v)$, where $c$ is the least frequent color in~$\delta(v)$.
\end{rrule}

\begin{lemma}
	\Cref{rr:redundant-color} is correct.
\end{lemma}
\begin{proof}
	Let $G'= (V, E')$ be the resulting graph after applying \Cref{rr:redundant-color}.
	If~$F \subseteq E'$ with~$|F| \ge k$ is stable in~$G'$, then it is also stable in~$G$.

	Conversely, suppose that $F \subseteq E$ with $|F| = k$ is stable in $G$.
	We show that there exists an edge set $F' \subseteq E'$ of size $k$ which is stable in $G'$.
	If $f_{F}(v) \ne c$, then $F$ is stable in $G$ since we have only deleted edges of color $c$ incident to $v$.
	Suppose that $f_{F}(v) = c$.
	Since there are at least~$2k$ other colors among the edges incident to~$v$,
	and $|\bigcup_{e \in F} e| \le 2k$,
	there exists a color~$c' \ne c$ such that the edges incident to~$v$ with color~$c'$ are vertex-disjoint to the edges in~$F$ (except for~$v$).
	As~$c$ is the least frequent color among the edges incident to~$v$, we obtain a desired solution~$F$ by replacing the edges of color~$c$ with those of color~$c'$.
\end{proof}

At this point, we can show that the resulting graph has at most $\bigO(k^3)$ edges.
Recall that $S$ is a vertex cover, i.e., every edge is incident to $S$.
Thus, there are at most $(2k + 1) \cdot |S| \in \bigO(k^2)$ colors that appear at least once.
By \Cref{rr:one-color}, each color appears at most $k$ times in the graph by \Cref{rr:one-color}, and consequently, we have $\bigO(k^3)$ edges.

We will show how to improve upon this using a meet-in-the-middle argument.
For this, we denote for a vertex~$v$ and color~$c$ the set of neighbors connected by edges of color~$c$ by~$N_c(v)$.
Intuitively speaking, if we have ``many'' vertices $v$ in $S$ such that there are ``many'' colors $c$ appearing ``many'' times in $\delta(v)$, then we can construct a stable set of size $k$.

\begin{rrule}
	\label{rr:meet-in-middle}
	Let $T \subseteq S$ be the set of vertices $v$ with $|\{ c \in \colors \mid |N_c(v) \setminus S| \ge 2k^{1/2} \}| \ge 2 k^{1/2}$.
	If $|T| \ge k^{1/2}$, then return \yes.
\end{rrule}

\begin{lemma}
	\label{lemma:meet-in-middle}
	\Cref{rr:meet-in-middle} is correct.
\end{lemma}
\begin{proof}
	Assume \wilog{} that $t \coloneqq |T| = k^{1/2}$ and let $T \coloneqq \{ v_1, \dots, v_{t} \}$.
	We show that the following greedy algorithm finds a stable set $F = F_{t}$ of size $k$:
	We begin with $F_0 := \emptyset$.
	We will construct $F_{q}$ for increasing $q \in [t]$ such that $|F_{q}| = \sum_{p=1}^q 2k^{1/2} - (p - 1)$.
	Let~$V_{F_{q-1}} \coloneqq (\bigcup_{e \in F_{q - 1}} e)\setminus S$ be the set of endpoints of~$F_{q-1}$ that are outside~$S$.
	For increasing $q \in [t]$ let $c \coloneqq \arg \max_{c' \in \colors} |N_{c'}(v_{q}) \setminus (S \cup V_{F_{q - 1}})|$.
	We claim that $|N_c(v_{q}) \setminus (S \cup V_{F_{q - 1}})| \ge 2k^{1/2} - (q - 1)$.
	Note that~$|V_{F_{q-1}}| \le |F_{q-1}| \le 2(q-1)k^{1/2}$.
	As there are $2 k^{1/2}$ colors~$c \in \colors$ such that~$|N_c(v_q) \setminus S| \ge 2k^{1/2}$,
	we have by pigeonhole principle that there exists a color~$c$ with $|N_c(v_{q}) \cap V_{F_{q - 1}}| \le |V_{F_{q-1}}|/(2k^{1/2}) = q-1$.
	Hence, $|N_c(v_q) \setminus S| - |N_c(v_q) \cap V_{F_{q-1}}| \ge 2k^{1/2}-(q-1)$.
	Now, let $F_{q}$ be formed by $F_{q - 1}$ and $2k^{1/2} - (q - 1)$ arbitrary edges between $v_{q}$ and $N_c(v_{q}) \setminus (S \cup V_{F_{q-1}})$.
	Finally, let $F := F_{t}$.
	We have $|F| = \sum_{q \in [t]} |F_{q}| - |F_{q - 1}| = \sum_{q \in [t]} 2k^{1/2} - (q - 1) \ge \sum_{q \in [t]} k^{1/2} = k$.
\end{proof}

After applying \Cref{rr:meet-in-middle}, we have ``not too many'' vertices such that there are ``many'' colors $c$ appearing ``many'' times in $\delta(v)$, which results in a smaller kernel.

\begin{theorem}
	\label{thm:kernel-stable-edges}
	\cec{} admits a kernel of size $\bigO(k^{5/2})$.
\end{theorem}
\begin{proof}
	We apply \Cref{rr:one-color,rr:matching,rr:redundant-color,rr:meet-in-middle} exhaustively and delete all isolated vertices.
	We show that the resulting graph has size $\bigO(k^{5/2})$.

	Let $T$ be as specified in \Cref{rr:meet-in-middle}.
	We can assume that $|T| \le k^{1/2}$ and thus there are at most $\bigO(k^2 \cdot |T|) = \bigO(k^{5/2})$ incident to $T$.
	Now consider a vertex $v \in S \setminus T$.
	Note that it suffices to bound the number of edges one of whose endpoint is in $S$, since there are at most $\bigO(k^2)$ edges with both endpoints in $S$.
	Let $X$ be the set of colors $c$ such that $|N_c(v) \setminus S| \ge 2 k^{1/2}$ and $Y$ be the set of colors $c \notin X$ such that $N_c(v) \setminus S \ne \emptyset$.
	By the definition of $T$, we have at most $2 k^{1/2}$ colors $c \in X$, amounting to at most $k \cdot |X| = \bigO(k^{3/2})$ edges incident to $v$ of color $c \in X$.
	Moreover, there are at most $\bigO(k^{3/2})$ edges incident to $v$ of color in $Y$ since at most $k^{1/2}$ edges of $\delta(v)$ have color $c \in Y$ and $|Y| \in \bigO(k)$.
	Thus, there are at most $\bigO(k^{5/2})$ vertices and edges in the graph since all isolated vertices have been deleted.
\end{proof}

This result generalizes to hypergraphs of order~$d$.
From \cref{rr:matching} we obtain a vertex cover of size at most~$dk$.
\cref{rr:one-color} translates immediately.
We can adapt \cref{rr:redundant-color} to bound the chromatic degree by~$dk$.
As observed above, this already yields a kernel.
\begin{corollary}
	\chc{} admits a kernel of size~$\bigO(d^2k^3)$.
\end{corollary}
Note that the meet-in-the middle argument does not translate to hypergraphs as in \cref{lemma:meet-in-middle} we implicitly assume each edge to have at most one endpoint outside $S$.

\section{Above-guarantee parameters}

There is a fairly straightforward FPT algorithm for \cec{} parameterized by~$r\coloneqq \abs{E} - k$, the number of permitted unstable edges.
One can simply choose an arbitrary pair of edges that share a vertex, but are colored differently and then branch into two cases depending on which of these two edges is unstable.
This yields an $\bigO^*(2^r)$-algorithm.
Cai and Leung~\cite{Cai2018} obtained a faster FPT algorithm with running time~$\bigO^*(1.2783^r)$ by reducing the problem to \textsc{Vertex Cover}.
The recently proposed $\bigO^*(1.25298^k)$ time algorithm for \textsc{Vertex Cover}~\cite{Harris2022} could be used to improve this running time.
There is a line of research that seeks to improve on such algorithms by introducing smaller parameters by considering a lower or upper bound and then using the difference between the traditional parameter and the lower or upper bound as a smaller parameter.
This approach has become known as ``above or below guarantee'' parameterization (see, e.g., \cite{Garg2016,Gutin2022,Kellerhals2022,Mahajan1999}).

One possible above-guarantee parameterization for \cec{} can be obtained by considering the so-called conflict graph of an edge-colored graph.
The conflict graph~$\partial(G,\ell)$, introduced by Angel et al.~\cite[Section~4]{Angel2016}, of a graph $G=(V,E)$ with an edge coloring $\ell \colon E \to \colors$ contains a vertex for every edge in $G$.
Two vertices in $\partial(G,\ell)$ are adjacent if the corresponding edges in $G$ have different colors and share a vertex.
Formally, $\partial(G, \ell) = (E, \partial E)$, where $\partial E = \{ \{ e, e' \} \subseteq E \mid e \cap e' \ne \emptyset \text{ and } \ell(e) \ne \ell(e') \}$.
Any stable edge set in $G$ corresponds to an independent set in $\partial(G,\ell)$ and consequently the minimum number of unstable edges in $G$ is the vertex cover number of $\partial(G,\ell)$.
It follows that any lower bound on the size of a minimum vertex cover in $\partial(G,\ell)$ is a lower bound on $r$.
\textsc{Vertex Cover} is \FPT{} when parameterized above the size of a maximum matching or above the minimum fractional solution of a standard linear program for \textsc{Vertex Cover}~\cite{Lokshtanov2014,Razgon2009}.
These results can also be used to solve \cec{} parameterized above the corresponding lower bounds on $\partial(G,\ell)$.

\subsection{Above degree-based lower bound}
\label{ssect:degree}
We will consider a different lower bound on the number $r$ of edges that must be unstable in any vertex coloring of a graph.
If we consider the edges incident to a particular vertex $v$, all of these edges save the edges of one color are unstable.
Hence, the number of unstable edges incident to $v$ is at least $\deg(v) - \max_{c\in C} \deg_c(v)$.
To obtain a lower bound on the total number of edges that must be unstable in a graph, we can add up this value over all vertices in the graph, but we must divide the sum by two to account for the fact that we may count edges twice in this way.
This lower bound can be formalized as follows.
For any graph~$G=(V,E)$, define
\begin{align*}
\rho(G,\ell) \coloneqq \frac{1}{2}\sum_{v \in V} \big(\!\deg(v) - \max_{c\in \colors} \deg_c(v)\big).
\end{align*}
Observe that for any \yes-instance $(G = (V, E), \ell, k = |E| - r)$ of \cec{}, we have $r \ge \rho(G,\ell)$.

Hence, in the following we will consider the parameterized complexity of \cec{} with respect to the parameter $r-\rho(G,\ell)$.
Unfortunately, it turns out that this problem is para-\NP-hard in general (\cref{thm:degree:hardness}).
However, our finding is that if we consider the parameter $r - \rho'(G, \ell)$ obtained from a smaller, yet still tight, lower bound, \cec{} becomes FPT:
\begin{align*}
	\rho'(G,\ell) \coloneqq \frac{1}{2}\sum_{v \in V} \min \big(\deg(v) - \max_{c\in C} \deg_c(v), \textstyle\frac{1}{2}\deg(v)\big).
\end{align*}
By definition, $\rho(G, \ell) \ge \rho(G', \ell)$ and hence it holds that $r \ge \rho'(G, \ell)$ for any \yes-instance of \cec{}.
We show that \cec{} is FPT for the parameter $r - \rho'(G, \ell)$ by proving that~$\rho'(G, \ell)$ is at most the optimal value of the LP relaxation of \textsc{Vertex Cover} on the conflict graph.
This allows us to use the FPT algorithm for \textsc{Vertex Cover} parameterized by the solution size minus the optimal value of the LP relaxation \cite{Lokshtanov2014}.

\begin{theorem}
	\label{thm:degree:fpt}
	\cec{} is \FPT{} with respect to $r-\rho'(G,\ell)$.
\end{theorem}
\begin{proof}
	\newcommand{\lp}{\alpha}
	\newcommand{\lpd}{\beta}
	We first show that $\lp(\partial(G, \ell)) \ge \rho'(G, \ell)$ for any graph $G = (V, E)$ with an edge coloring $\ell \colon V \to C$, where $\lp(H)$ is the relaxed optimum of \textsc{Vertex Cover} on a graph $H = (V', E')$.
	More precisely, we define
	\begin{align*}
		\lp(H) := \min \sum_{v \in V'} x_v
		\quad\text{ subject to }
		&\quad x_u + x_{v} \ge 1 \quad \forall \{ u, v \} \in E' \\[-1ex]
    	&\quad 0 \le x_v \le 1 \quad \forall v \in V'.
	\end{align*}
	Note that if $x_v \in \{0, 1\}$ for each~$v \in V'$, then this formulation finds a minimum vertex cover in~$H$.
	To show that $\lp(\partial(G, \ell)) \ge \rho'(G, \ell)$, we consider the dual:
	\begin{align*}
		\lpd(H) := \max \sum_{e \in E'} y_{e}
		\quad\text{ subject to }
		&\quad \sum_{e \in E'\!, v \in e} y_{e} \le 1 \quad \forall v \in V' \\[-.5ex]
    	&\quad 0 \le y_{e} \le 1 \quad \forall e \in E'.
	\end{align*}
	By the strong duality theorem, we have $\lp(H) = \lpd(H)$.
	So, it suffices to show that $\lpd(\partial(G, \ell)) \ge \rho'(G, \ell)$.
	To that end, we give a feasible solution $\{ y_{\partial e} \}_{\partial e \in \partial E}$ to the dual LP whose value is at least $\rho'(G, \ell)$.
	To construct such a solution, we show the following.
	\begin{claim}
		\label{lemma:bc-lp}
		Let~$H$ be a complete~$\kappa$-partite graph whose~$i$-th part has~$n_i$ vertices, wherein~$n_i \ge n_j$ for~$i < j \in [\kappa]$.
		If~$H$ has at least two vertices, then
		\begin{align*}
			\lpd(H) \ge \begin{cases}
				\sum_{i=2}^{\kappa} n_i & \text{if } n_1 \ge \sum_{i=2}^{\kappa} n_i. \\
				\frac{1}{2} \sum_{i=1}^{\kappa} n_i & \text{otherwise}.
			\end{cases}
		\end{align*}
	\end{claim}
	\begin{proof}
		If $n_1 \ge \sum_{i=2}^{\kappa} n_i$, then consider a solution to the dual LP defined as follows:
		For every vertex $v$ in the $i$-th part for $i > 1$, we choose an arbitrary distinct vertex $u$ from the first part and let $y_e = 1$ for edge $e=\{u,v\}$ and let $y_e = 0$ for all other edges $e$.
		Note that this yields a (valid) solution of value $\sum_{i=2}^{\kappa} n_i$.
		
		Suppose that $n_1 \le \sum_{i=2}^{\kappa} n_i$.
		If $H$ has at most three vertices, then we have $n_1 = n_2 = 1$ and $n_3 \in \{ 0, 1 \}$, that is, a complete graph on two or three vertices.
		Then, let $y_e = 1/2$ for every edge $e$ to obtain a solution that has value $\frac{1}{2} \sum_{i=1}^{\kappa} n_i$.
		If $H$ has more than three vertices, then choose a vertex~$u$ from the first part and a vertex~$v$ from the second part. %
		Let $y_{\{ u, v \}} = 1$ and let $y_{e} = 0$ for every edge $e$ incident to exactly one of $u$ and $v$.
		After deleting $u$ and $v$ from the graph, we end up with a complete $\kappa'$-partite graph $H'$ for which $n_1' < \sum_{i=2}^{\kappa'} n_i'$ holds, where $n_i'$ is the size of the $i$-th largest part:
		If $n_1' = n_1 - 1$, then $\sum_{i=2}^{\kappa'} n_i' = (\sum_{i=1}^{\kappa'} n_i') - n_1' = (\sum_{i=1}^{\kappa} n_i) - 2 - (n_1 - 1)  = (\sum_{i=2}^{\kappa} n_i) - 1 \ge n_1 - 1 = n_1'$.
		Otherwise, we have $n_1' = n_1$ and hence $n_1 = n_2 = n_3$.
		If  $n_1 = 1$, then we have $\sum_{i=2}^{\kappa'} n_i' = (\sum_{i=1}^{\kappa'} n_i') - n_1' \ge (\sum_{i=1}^{\kappa} n_i) - 2 - 1 \ge 1 = n_1'$, since $H$ has at least four vertices.
		If $n_1 \ge 2$, then we have $\sum_{i=2}^{\kappa} n_i' \ge n_2' + n_3' \ge (n_1 - 1) + (n_2 - 1) = 2n_1 - 2 \ge n_1 = n_1'$.
		Thus, an inductive argument shows that there is a solution whose value is at least $\frac{1}{2} \sum_{i=1}^{\kappa} n_i$.
	\end{proof}
	We use \Cref{lemma:bc-lp} to obtain a solution $\{ y_{\partial e} \}_{\partial e \in \partial E}$ whose value is at least $\rho'(G, \ell)$.
	Initially, let $y_{\partial e} = 0$ for every $\partial e \in \partial E$.
	For every vertex $v \in V$, we do the following:
	The subgraph of $\partial(G, \ell)$ induced by the set $\delta(v)$ of edges incident to $v$ is a complete $\kappa$-partite graph, where~$\kappa$ is the number of colors incident to~$v$.
	Thus, by \Cref{lemma:bc-lp}, it has a dual LP solution $\{ z_{\partial e}^v \}$ of value at least $\min(\deg(v) - \max_{c \in C} \deg_c(v), \frac{1}{2} \deg(v))$.
	For every edge $\partial e \in \partial E$ in this induced subgraph, let $y_{\partial e} = \frac{1}{2} z_{\partial e}^v$.
	We thus obtain a solution of value at least $\rho'(G, \ell)$.
	Moreover, this solution is feasible because for every edge $e = \{ u, v \}$,  we have
	\begin{equation}
		\label{eq:solution-feasible}
		\begin{aligned}
			\sum_{\substack{\partial e \in \partial E\\ e \in \partial e}} y_{\partial e} = \sum_{\substack{\{ e, e' \} \in \partial e\\ u \in e'}} y_{\partial e} + \sum_{\substack{\{ e, e'\} \in \partial e\\ v \in e'}} y_{\partial e} = \sum_{\substack{\{ e, e' \} \in \partial e\\ u \in e'}} {\textstyle\frac{1}{2}} z_{\partial e}^u + \sum_{\substack{\{ e, e'\} \in \partial e\\ v \in e'}} \textstyle{\frac{1}{2}} z_{\partial e}^v \le 1.
		\end{aligned}
	\end{equation}
		
	Since \textsc{Vertex Cover} is FPT when parameterized by the solution size minus the LP relaxation value \cite{Lokshtanov2014}, it follows that \cec{} is \FPT{} with respect to $r - \alpha(H)$, and thereby $r - \rho'(G, \ell)$.
\end{proof}

We remark that this leads to an FPT algorithm for the smaller parameter $r - \rho(G, \ell)$ if the maximum chromatic degree is two.
Note that \cec{} remains \NP-hard with this restriction \cite{Angel2016}.

\begin{corollary}
	\cec{} is FPT with respect to $r - \rho(G, \ell)$ if the chromatic degree of every vertex is at most two.
\end{corollary}
\begin{proof}
	Since for each vertex $v \in V$ that is incident to edges of at most two colors $c, c' \in C$ (possibly $\deg_c(v) = 0$ or $\deg_{c'}(v) = 0$), we have $\deg(v) = \deg_c(G) + \deg_{c'}(G)$, we have $\deg(v) - \max_{c \in C} \deg_{c}(v) = \min_{c \in C} \deg_c(v) \le \frac{1}{2}(\max_{c \in C} \deg_c(v) + \min_{c \in C} \deg_c(v)) = \deg(v)$.
	Thus, we have $\min(\deg(v) - \max_{c \in C} \deg_c(v), \allowbreak \frac{1}{2} \deg(v)) = \deg(v) - \max_{c \in C} \deg_c(v)$ for every vertex $v \in V$.
	It follows that $\rho(G, \ell) = \rho'(G, \ell)$.
	Thus, we obtain an FPT algorithm by \cref{thm:degree:fpt}.
\end{proof}

We complement our positive results by showing that when each vertex may be incident to edges of three (or more) different colors, the problem becomes para-\NP-hard with respect to~$r - \rho(G, \ell)$.

\begin{theorem}
	\label{thm:degree:hardness}
	\cec{} is \NP-hard even if $r-\rho(G,\ell) = 0$ and the chromatic degree of every vertex is at most three.
\end{theorem}
\begin{proof}
	We reduce from the NP-hard \textsc{Monotone One-in-Three SAT} problem defined as follows:
	The input consists of a set of variables $X$ and a set of clauses $Y$ consisting of three positive literals.
	The task is to find a truth assignment such that exactly one literal of every clause is satisfied.
	It is known that this problem is NP-hard even if every variable appears exactly three times \cite{Moore2001}.
	
	Let $(X, Y)$ be an instance of \textsc{Monotone One-in-Three SAT} on $n$ variables and $n$ clauses.
	We construct an instance of \cec{} as follows.
	Our reduction is inspired by the trick of \citet{Wu2011}.
	For every clause $y \in Y$, we introduce a vertex $v_y$.
	For every variable $x$ appearing in three clauses $y_1, y_2, y_3$, we introduce six vertices $u_x^1, u_x^2, u_x^3, w_x^1, w_x^2, w_x^3$ and add edges so that $\{ u_x^1, u_x^2, u_x^3 \}$, $\{ w_x^1, w_x^2, w_x^3 \}$, and $\{ y_i, u_x^i, w_x^i \}$ for each $i \in [3]$ form triangles in the resulting graph.
	(Essentially, the variable $x$ is true if and only if the three triangles $\{ y_i, u_x^i, w_x^i \}$ are part of a stable set.)
	For every triangle, we color its edges in a color unique to this triangle.
	This concludes the construction of $G = (V, E)$ and $\ell$.
	Note that $G$ has $6n$ vertices and $15n$ edges.
	Let $k := 7n$ (i.e., $r = 8n$).
	Note that every vertex $v_y$ has chromatic degree three and $u_x^i$ and $w_x^i$ have chromatic degree two.
	Since $\deg(y_v) - \max_{c \in C} \deg_c(y_v) = 6 - 2 = 4$ and $\deg(u_x^i) - \max_{c \in C} \deg_c(u_x^i) = \deg(w_x^i) - \max_{c \in C} \deg_c(w_x^i) = 4 - 2 = 2$, we have $\rho(G, \ell) = \frac{1}{2} (\sum_{v \in V} \deg(v) - \max_{c \in C} \deg_c(v)) = \frac{1}{2}(4n + 2 \cdot 6n) = 8n = r$.
	It remains to show the equivalence between $(X, Y)$ and $(G, \ell, k)$.

	\RD{}
	Suppose that $(X, Y)$ admits a truth assignment satisfying exactly one literal of each clause.
	We will construct a stable set $F$ of size $k = 7n$.
	Suppose that a variable $x$ appears in clauses $y_1$, $y_2$, and $y_3$.
	If $x$ is true in the assignment, we add the edges of three triangles $\{ y_i, u_x^i, w_x^i \}$ to $F$.
	Otherwise, we add the edges of triangles $\{ u_x^1, u_x^2, u_x^3 \}$ and $\{ w_x^1, w_x^2, w_x^3 \}$.
	Since every clause is satisfied exactly once, $F$ is stable.
	A simple counting argument shows that $\frac{1}{3} n$ variables are true. Thus, $F$ has size $9 \cdot \frac{1}{3} n + 6 \cdot \frac{2}{3} n = 7n = k$.

	\LD{}
	Suppose that there is a stable set $F \subseteq E$ of size $k = 7n$.
	Since every vertex is incident to at most two edges of the same color, any stable set contains at most $\frac{1}{2} \cdot 2|V| = 7n$ edges (note that $\frac{1}{2}$ is necessary to make up for the double counting).
	It follows that every vertex (in particular, each of $u_x^1, u_x^2, u_x^3, w_x^1, w_x^2, w_x^3$) is incident to exactly two edges of $F$.
	This only holds true if for each $x \in X$, one of the following holds:
	all edges of three triangles $\{ y_i, u_x^i, w_x^i \}$ belong to $F$ or all edges of two triangles $\{ u_x^1, u_x^2, u_x^3 \}$ and $\{ w_x^1, w_x^2, w_x^3 \}$ belong to $F$.
	The truth assignment in which $x$ is true if and only if the edges of three triangles $\{ y_i, u_x^i, w_x^i \}$ are part of $F$ satisfies exactly one literal of each clause.
\end{proof}

\paragraph{Hypergraphs.}
We can lift the algorithm in \cref{thm:degree:fpt} to hypergraphs of order $d$, albeit with a smaller lower bound
\begin{align*}
	\rho_H(G, \ell) \coloneqq \frac{1}{d} \sum_{v \in V} \min \big(\deg(v) - \max_{c \in C} \deg_c(v), \textstyle\frac{1}{2}\deg(v) \big).
\end{align*}
As in \eqref{eq:solution-feasible} we then have at most $d$ summands (one for each endpoint of hyperedge~$e$), the solution remains feasible.

\begin{corollary}
	\chc{} is FPT with respect to~$r - \rho_H(G, \ell)$.
\end{corollary}

\subsection{Above matching-based lower bounds}
We will now consider matchings as lower bounds on $k$.
A \emph{matching} in a graph $G=(V,E)$ is a set of edges $M\subseteq E$ such that $e\cap e' \neq \emptyset$ for any two distinct $e,e'\in M$.
Let $M(G)$ denote the size of a maximum matching in a graph $G$.
As we noted in \cref{rr:matching}, any matching is trivially stable, which implies the following:

\begin{observation}
	If $M(G) \geq k$, then $(G=(V,E),\ell,k)$ is a \yes-instance.
\end{observation}

A matching $M$ in $G=(V,E)$ is \emph{induced} if there are no $e,e'\in M$ and $v\in e, v'\in e'$ such that $v$ and $v'$ are adjacent.
Let $I(G)$ denote the size of a maximum induced matching in a graph $G$.
Of course, $I(G) \leq M(G)$, implying that if $I(G) \geq k$, then $(G=(V,E),\ell,k)$ is a \yes-instance.
In the following we will consider the parameters $k-M(G)$ and $k-I(G)$.
Finding a maximum induced matching is \NP-hard~\cite{Stockmeyer1982}, so we assume that such a matching is given as part of the input, when we deal with the latter parameter.

Our main result concerning parameterizations of \cec{} above matching-based lower bounds is that the problem is \FPT{} with respect to $(k - I(G)) + |C|$ (\cref{thm:FPT-above-IM}).
We also show that the problem is \XP{} with respect to $k-I(G)$, that is, \cecS{} can be solved in polynomial time for any constant value of $k-I(G)$ (\cref{thm:XP-above-IM}).
We will also show that these are, in a sense, the best one can do concerning matching-based lower bounds.
We will prove that \cec{} is \NP-hard for $k-M(G) = 1$ and $|C| = 19$ (\cref{thm:paraNPh-above-matching}), implying that there is not even \XP{} algorithm for the parameter $(k - M(G)) + |C|$, unless $\text{P} = \NP$.
Moreover, we prove that \cecS{} is \W{1}-hard with respect to $k-I(G)$ (\cref{thm:w1h-above-IM}), meaning that there is no \FPT{} algorithm for $k-I(G)$, unless $\FPT = \W{1}$.

In this section, we write $\deg_F(e) \coloneqq \deg_F(u) + \deg_F(v)$ for an edge $e = \{ u , v \}$.
We will use the following observation throughout this section.
\begin{observation}
	\label{obs:matching-counting}
	Let $G = (V, E)$ be a graph with a matching $M$.
	For an edge subset $F \subseteq E$, it holds that $2|F| = \sum_{v \in V} \deg_{F}(v) = \sum_{v \in X} \deg_F(v) +  \sum_{e \in M} \deg_F(e)$, where $X := V \setminus \bigcup_{e \in M} e$ is the set of unmatched vertices. 
\end{observation}

We will start by presenting the \FPT{} algorithm for $(k-I(G)) + |C|$.

\begin{theorem}
	\label{thm:FPT-above-IM}
	Given a maximum induced matching, \cec{} can be solved in $\bigO^*(|C|^{\bigO(k - I(G))})$ time.
\end{theorem}
\begin{proof}
	Suppose that we are a given \yes-instance $(G = (V, E), \ell, k)$ of \cec{} along with an induced matching $M$.
	Among possibly many solutions, our goal is to find a stable set $F$ of size $k$ such that $|F \setminus M| \le |F' \setminus M|$ for any stable set of size $k$ (in other words, $F$ contains as many edges of $M$ as possible).
	This restriction on $F$ will play a central role in the algorithm.
	We will first show that $|F \setminus M| \le 2(k - |M|)$.

	We first show that each edge~$e \in M \setminus F$ intersects at least two edges in~$F \setminus M$.
	Suppose not, that is, there is an edge~$e \in M \setminus F$ that intersects at most one edge in~$F$.
	Let~$f$ be that edge, or let~$f \in F$ be arbitrary if~$e$ intersects no edge in~$F$.
	Then~$F' \coloneqq (F \setminus \{f\}) \cup \{e\}$ is stable and~$|F' \setminus M| > |F \setminus M|$, a contradiction.

	So assume that each edge~$e \in M \setminus F$ intersects at least two edges in~$F \setminus M$.
	As~$M$ is an induced matching, every edge in~$F \setminus M$ intersects at most one edge in~$M$.
	Thus,
	$|F \setminus M| \ge 2|M \setminus F|$, which yields
	\begin{align*}
		|F \setminus M| 
		&\le |F \setminus M| + (|F \setminus M| - 2|M \setminus F|) \\ 
		&= 2(|F| - |F \cap M| - |M \setminus F|) = 2(k - |M|).
	\end{align*}

	To obtain an FPT algorithm, we use the color coding technique.
	We say that a vertex coloring $f\colon V \to \colors$ is \emph{good} for $F$ if $F\setminus M$ is stable under $f$.
	We color the vertices of $G$ independently and uniformly at random, that is we assign color $c\in \colors$ to each vertex $V$ with probability~$\abs{C}^{-1}$.
	Then, the vertex $v\in \bigcup_{e\in F\setminus M} e$ receives the color $f_F(v)$ with probability~$\abs{C}^{-1}$.
	This implies that the probability that $f$ is good is at least $\abs{C}^{-\abs{\bigcup_{e\in F\setminus M} e}} = \abs{C}^{-2\abs{F\setminus M}} \geq \abs{C}^{-4(k - \abs{M})}$.
	Hence, we obtain a good coloring of $G$ with constant probability by repeating the coloring $\abs{C}^{\bigO(k-\abs{M})}$ times.
	Given a good coloring $f$ of $G$, we obtain a coloring $f'$ with at least $k$ stable edges in polynomial time as follows.
	For every $e=\{u,v\} \in M$, we check whether at least two edges in $\delta(u) \cup \delta(v)$ are stable under $f$.
	If so, let $f'(u) \coloneqq f(u)$ and $f'(v) \coloneqq f(v)$.
	If not, set $f'(u) \coloneqq f'(v) \coloneqq \ell(e)$.
	
	We must show that, if $G$ contains a stable edge set of size at least $k$, then algorithm finds a coloring $f'$ with at least $k$ stable edges with constant probability.
	Let $F$ be a stable edge set of size at least $k$ that maximizes the number of edges in $M$.
	We will show that, if $f$ is good for $F$, then $F\setminus M$ is stable under $f'$ and $\abs{\{e \in E\setminus (M \cup F) \mid e \text{ is stable under } f'\}} \geq \abs{\{e \in M \cap F \mid e \text{ is unstable under } f'\}}$.
	This implies that at least $\abs{F} = k$ edges are stable under $f'$.
	First, every $e\in F\setminus M$ is stable under $f$.
	Moreover, by the observation above, if $e$ intersects $e' \in M \setminus F$, then there must be a second edge $e'' \in F \setminus F$.
	The edge $e''$ must also be stable under $f$.
	It follows that $f$ and $f'$ agree on the endpoints of $e$ and, hence, $e$ is stable under $f'$.
	Now consider $e = \{u,v\} \in M \cap F$.
	If $e$ is unstable under $f'$, then there must be at least two edges $e',e'' \in \delta(u) \cup \delta(v)$ that are stable under $f$.
	Because $M$ is an induced matching, it follows that $e$ is the only edge in $M$ that intersects $e'$ or $e''$.
	Hence, $f$ and $f'$ agree on the endpoints of $e'$ and $e''$, implying that these two edges are stable under $f'$.
	Hence, for each $e \in M \cap F$ that is unstable under $f'$, there are at least two edges in $E\setminus (M \cup F)$ that are stable under $f'$.
\end{proof}

The bound $|F \setminus M| \le 2(k - |M|)$ in the proof of \cref{thm:FPT-above-IM} actually implies that \cec{} is XP with respect to the parameter $k - |M|$:

\begin{corollary}
	\label{thm:XP-above-IM}
	Given a maximum induced matching $M$, \cec{} is can be solved in $n^{\bigO(k - I(G))}$ time.
\end{corollary}
\begin{proof}
	Let $F$ be a stable set as defined in the proof of \cref{thm:FPT-above-IM}.
	Since $|F \setminus M| \le 2(k - |M|)$, we can guess the set~$F' \coloneqq F \setminus M$ in $n^{\bigO(k-|M|)}$ time.
	Then what remains is to decide whether~$F'$ can be extended to a stable set~$F$ with~$k - |F'|$ edges of the matching~$M$.
	As this can be done in polynomial time, the theorem follows.
\end{proof}

Observe that the above definition of induced matchings carries over to hypergraphs.
The bound on~$|F \setminus M|$ also holds on hypergraphs of order $d$, and the two results above carry over.
As~$|\bigcup_{e \in F\setminus M}e| \le d |F \setminus M| \le 2d(k - |M|)$ in hypergraphs, the probability that the coloring~$f$ is good is~$|C|^{-2d(k-|M|)}$ and we have to adjust the number of repetitions accordingly.
\begin{corollary}
	Given an induced matching~$M$ of size~$I(G)$, \chc{} can be solved in~$\bigO^*(|C|^{\bigO(d(k-I(G)))})$ time and in~$n^{\bigO(k-I(G))}$ time.
\end{corollary}

Next, we will show that this \XP{} algorithm for $k-I(G)$ cannot be improved to an \XP{} algorithm for $k-M(G)$, unless $\text{P} = \NP$.
This suggests that the assumption that there is no edge connecting endpoints of two edges of $M$ is imperative in the XP algorithm of \cref{thm:XP-above-IM}.
We will show that \cec{} is \NP-hard when $k-M(G) =1$ and $\abs{C} = 19$ by a reduction from the \NP-hard~\cite{Wu2012} problem \textsc{Two Disjoint Monochromatic Paths}, in which we are given an edge-bicolored graph and two terminal pairs and are asked to find two vertex-disjoint monochromatic paths to connect each pair.
The reduction uses \cref{obs:matching-counting} and is based on the following idea:
We create a graph with a perfect matching in which there are two matching edges which can have degree three in a solution while all other matching edges can have degree at most two.
A solution must then connect those two particular matching edges with a kind of path that can be translated into two disjoint monochromatic paths in the input instance.

\begin{theorem}
	\label{thm:paraNPh-above-matching}
	\cec{} is \NP-hard even if $k-M(G) = 1$ and $\abs{C}\leq 19$.
\end{theorem}
\begin{proof}
	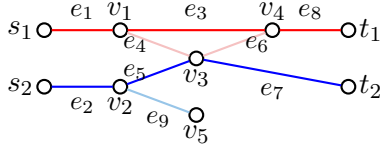
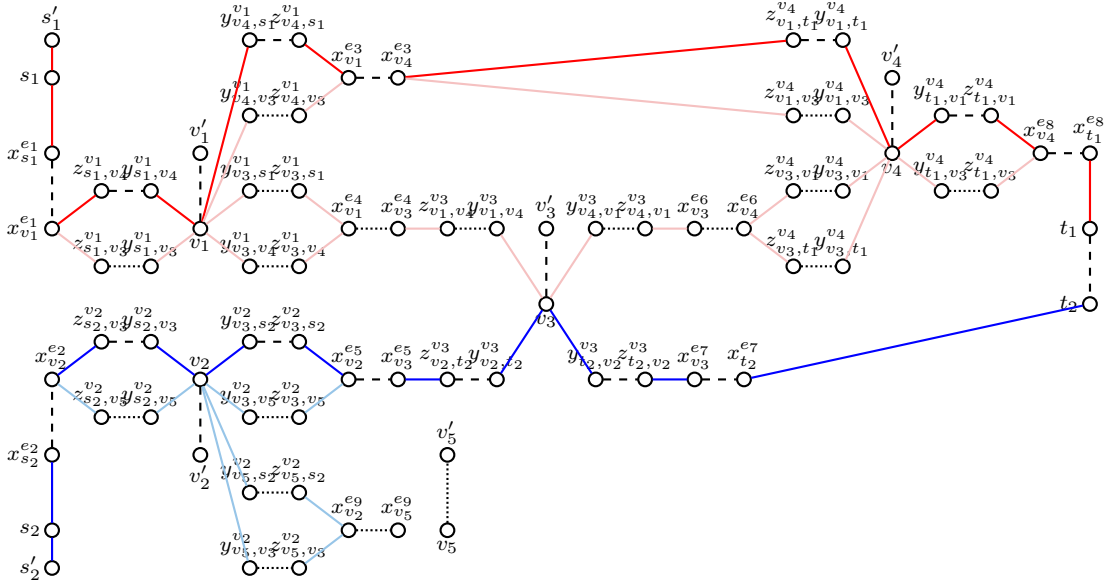
\begin{figure}[t]
		\centering
		\begin{subfigure}{\textwidth}
			\centering
			\begin{tikzpicture}[rotate=90,xscale=-0.75,yscale=-1]
				\tikzpreamble
				\tikzstyle{every path}=[thick]

				\node[xnode] (s1) at (0,0) {};
				\node[left] at (s1) {$s_1$};
				\node[xnode] (s2) at (1,0) {};
				\node[left] at (s2) {$s_2$};
				\node[xnode] (v1) at (0,1) {};
				\node[above] at (v1) {$v_1$};
				\node[xnode] (v2) at (1,1) {};
				\node[below] at (v2) {$v_2$};
				\node[xnode] (v3) at (0.5,2) {};
				\node[below] at (v3) {$v_3$};
				\node[xnode] (v4) at (0,3) {};
				\node[above] at (v4) {$v_4$};
				\node[xnode] (t1) at (0,4) {};
				\node[right] at (t1) {$t_1$};
				\node[xnode] (t2) at (1,4) {};
				\node[right] at (t2) {$t_2$};
				\node[xnode] (v5) at (1.5,2) {};
				\node[below] at (v5) {$v_5$};

				\draw[red] (s1) edge node[above,black] {\small $e_1$}  (v1);
				\draw[blue] (s2) edge  node[below,black] {\small $e_2$} (v2);
				\draw[red] (v1) edge  node[above,black] {\small $e_3$} (v4);
				\draw[lred] (v1) edge  node[left,black] {\small $e_4$} (v3);
				\draw[blue] (v2) edge  node[left,black] {\small $e_5$} (v3);
				\draw[lred] (v3) edge  node[right,black] {\small $e_6$} (v4);
				\draw[blue] (v3) edge  node[below,black] {\small $e_7$} (t2);
				\draw[red] (v4) edge  node[above,black] {\small $e_8$} (t1);
				\draw[lblue] (v2) edge  node[below,black] {\small $e_9$} (v5);
 			\end{tikzpicture}
 			\caption{An instance of \textsc{Two Disjoint Monochromatic Paths}.
 			The dark red and dark blue edges form disjoint monochromatic paths.
 			The light red and light blue edges are not part of this solution.}
		\end{subfigure}
		\hfill
		\begin{subfigure}{\textwidth}
			\centering
			\begin{tikzpicture}[xscale=0.65]
				\tikzpreamble
				 \tikzstyle{every node}=[font=\scriptsize]
				\tikzstyle{every path}=[thick]
				
				\node[xnode] (s'1) at (3,4.5) {};
				\node[above] at (s'1) {$s'_1$};
				\node[xnode] (s1) at (3,4) {};
				\node[left] at (s1) {$s_1$};
				\node[xnode] (xe1s1) at (3,3) {};
				\node[left] at (xe1s1) {$x^{e_1}_{s_1}$};
				\node[xnode] (xe1v1) at (3,2) {};
				\node[left] at (xe1v1) {$x^{e_1}_{v_1}$};
				\node[xnode] (zv1s1v4) at (4,2.5) {};
				\node[above] at (zv1s1v4) {$z^{v_1}_{s_1,v_4}$};
				\node[xnode] (zv1s1v3) at (4,1.5) {};
				\node[above] at (zv1s1v3) {$z^{v_1}_{s_1,v_3}$};
				\node[xnode] (yv1s1v4) at (5,2.5) {};
				\node[above] at (yv1s1v4) {$y^{v_1}_{s_1,v_4}$};
				\node[xnode] (yv1s1v3) at (5,1.5) {};
				\node[above] at (yv1s1v3) {$y^{v_1}_{s_1,v_3}$};
				
				\node[xnode] (v1) at (6,2) {};
				\node[below] at (v1) {$v_1$};
				\node[xnode] (v'1) at (6,3) {};
				\node[above] at (v'1) {$v'_1$};
				
				\draw[red] (s'1) -- (s1);
				\draw[red] (s1) -- (xe1s1);
				\draw[dashed] (xe1s1) -- (xe1v1);
				\draw[red] (xe1v1) -- (zv1s1v4);
				\draw[lred] (xe1v1) -- (zv1s1v3);
				\draw[dashed] (zv1s1v4) -- (yv1s1v4);
				\draw[densely dotted] (zv1s1v3) -- (yv1s1v3);
				\draw[red] (yv1s1v4) -- (v1);
				\draw[lred] (yv1s1v3) -- (v1);
				\draw[dashed] (v1) -- (v'1);

				\node[xnode] (s'2) at (3,-2.5) {};
				\node[left] at (s'2) {$s'_2$};
				\node[xnode] (s2) at (3,-2) {};
				\node[left] at (s2) {$s_2$};
				\node[xnode] (xe2s2) at (3,-1) {};
				\node[left] at (xe2s2) {$x^{e_2}_{s_2}$};
				\node[xnode] (xe2v2) at (3,0) {};
				\node[above] at (xe2v2) {$x^{e_2}_{v_2}$};
				\node[xnode] (zv2s2v3) at (4,0.5) {};
				\node[above] at (zv2s2v3) {$z^{v_2}_{s_2,v_3}$};
				\node[xnode] (zv2s2v5) at (4,-0.5) {};
				\node[above] at (zv2s2v5) {$z^{v_2}_{s_2,v_5}$};
				\node[xnode] (yv2s2v3) at (5,0.5) {};
				\node[above] at (yv2s2v3) {$y^{v_2}_{s_2,v_3}$};
				\node[xnode] (yv2s2v5) at (5,-0.5) {};
				\node[above] at (yv2s2v5) {$y^{v_2}_{s_2,v_5}$};
				
				\node[xnode] (v2) at (6,0) {};
				\node[above] at (v2) {$v_2$};
				\node[xnode] (v'2) at (6,-1) {};
				\node[below] at (v'2) {$v'_2$};
				
				\draw[blue] (s'2) -- (s2);
				\draw[blue] (s2) -- (xe2s2);
				\draw[dashed] (xe2s2) -- (xe2v2);
				\draw[blue] (xe2v2) -- (zv2s2v3);
				\draw[lblue] (xe2v2) -- (zv2s2v5);
				\draw[dashed] (zv2s2v3) -- (yv2s2v3);
				\draw[densely dotted] (zv2s2v5) -- (yv2s2v5);
				\draw[lblue] (yv2s2v5) -- (v2);
				\draw[blue] (yv2s2v3) -- (v2);
				\draw[dashed] (v2) -- (v'2);
				
				\node[xnode] (yv1v4s1) at (7,4.5) {};
				\node[above] at (yv1v4s1) {$y^{v_1}_{v_4,s_1}$};
				\node[xnode] (yv1v4v3) at (7,3.5) {};
				\node[above] at (yv1v4v3) {$y^{v_1}_{v_4,v_3}$};
				\node[xnode] (yv1v3s1) at (7,2.5) {};
				\node[above] at (yv1v3s1) {$y^{v_1}_{v_3,s_1}$};
				\node[xnode] (yv1v3v4) at (7,1.5) {};
				\node[above] at (yv1v3v4) {$y^{v_1}_{v_3,v_4}$};
				\node[xnode] (zv1v4s1) at (8,4.5) {};
				\node[above] at (zv1v4s1) {$z^{v_1}_{v_4,s_1}$};
				\node[xnode] (zv1v4v3) at (8,3.5) {};
				\node[above] at (zv1v4v3) {$z^{v_1}_{v_4,v_3}$};
				\node[xnode] (zv1v3s1) at (8,2.5) {};
				\node[above] at (zv1v3s1) {$z^{v_1}_{v_3,s_1}$};
				\node[xnode] (zv1v3v4) at (8,1.5) {};
				\node[above] at (zv1v3v4) {$z^{v_1}_{v_3,v_4}$};
				
				\draw[red] (v1) -- (yv1v4s1);
				\draw[lred] (v1) -- (yv1v3s1);
				\draw[lred] (v1) -- (yv1v4v3);
				\draw[lred] (v1) -- (yv1v3v4);
				\draw[dashed] (yv1v4s1) -- (zv1v4s1);
				\draw[densely dotted] (yv1v3s1) -- (zv1v3s1);
				\draw[densely dotted] (yv1v4v3) -- (zv1v4v3);
				\draw[densely dotted] (yv1v3v4) -- (zv1v3v4);
				
				\node[xnode] (yv2v3s2) at (7,0.5) {};
				\node[above] at (yv2v3s2) {$y^{v_2}_{v_3,s_2}$};
				\node[xnode] (yv2v3v5) at (7,-0.5) {};
				\node[above] at (yv2v3v5) {$y^{v_2}_{v_3,v_5}$};
				\node[xnode] (yv2v5s2) at (7,-1.5) {};
				\node[above] at (yv2v5s2) {$y^{v_2}_{v_5,s_2}$};
				\node[xnode] (yv2v5v3) at (7,-2.5) {};
				\node[above] at (yv2v5v3) {$y^{v_2}_{v_5,v_3}$};
				\node[xnode] (zv2v3s2) at (8,0.5) {};
				\node[above] at (zv2v3s2) {$z^{v_2}_{v_3,s_2}$};
				\node[xnode] (zv2v3v5) at (8,-0.5) {};
				\node[above] at (zv2v3v5) {$z^{v_2}_{v_3,v_5}$};
				\node[xnode] (zv2v5s2) at (8,-1.5) {};
				\node[above] at (zv2v5s2) {$z^{v_2}_{v_5,s_2}$};
				\node[xnode] (zv2v5v3) at (8,-2.5) {};
				\node[above] at (zv2v5v3) {$z^{v_2}_{v_5,v_3}$};
				
				\draw[blue] (v2) -- (yv2v3s2);
				\draw[lblue] (v2) -- (yv2v3v5);
				\draw[lblue] (v2) -- (yv2v5s2);
				\draw[lblue] (v2) -- (yv2v5v3);
				\draw[dashed] (yv2v3s2) -- (zv2v3s2);
				\draw[densely dotted] (yv2v3v5) -- (zv2v3v5);
				\draw[densely dotted] (yv2v5s2) -- (zv2v5s2);
				\draw[densely dotted] (yv2v5v3) -- (zv2v5v3);
				
				\node[xnode] (xe3v1) at (9,4) {};
				\node[above] at (xe3v1) {$x^{e_3}_{v_1}$};
				\node[xnode] (xe4v1) at (9,2) {};
				\node[above] at (xe4v1) {$x^{e_4}_{v_1}$};
				\node[xnode] (xe3v4) at (10,4) {};
				\node[above] at (xe3v4) {$x^{e_3}_{v_4}$};
				\node[xnode] (xe4v3) at (10,2) {};
				\node[above] at (xe4v3) {$x^{e_4}_{v_3}$};
				
				\draw[red] (zv1v4s1) --(xe3v1);
				\draw[lred] (zv1v4v3) --(xe3v1);
				\draw[lred] (zv1v3s1) -- (xe4v1);
				\draw[lred] (zv1v3v4) -- (xe4v1);
				\draw[dashed] (xe3v1) -- (xe3v4);
				\draw[densely dotted] (xe4v1) -- (xe4v3);
				
				\node[xnode] (xe5v2) at (9,0) {};
				\node[above] at (xe5v2) {$x^{e_5}_{v_2}$};
				\node[xnode] (xe9v2) at (9,-2) {};
				\node[above] at (xe9v2) {$x^{e_9}_{v_2}$};
				\node[xnode] (xe5v3) at (10,0) {};
				\node[above] at (xe5v3) {$x^{e_5}_{v_3}$};
				\node[xnode] (xe9v5) at (10,-2) {};
				\node[above] at (xe9v5) {$x^{e_9}_{v_5}$};
				
				\draw[blue] (zv2v3s2) -- (xe5v2);
				\draw[lblue] (zv2v3v5) -- (xe5v2);
				\draw[lblue] (zv2v5s2) -- (xe9v2);
				\draw[lblue] (zv2v5v3) -- (xe9v2);
				\draw[dashed] (xe5v2) -- (xe5v3);
				\draw[densely dotted] (xe9v2) -- (xe9v5);
				
				\node[xnode] (zv3v1v4) at (11,2) {};
				\node[above] at (zv3v1v4) {$z^{v_3}_{v_1,v_4}$};
				\node[xnode] (yv3v1v4) at (12,2) {};
				\node[above] at (yv3v1v4) {$y^{v_3}_{v_1,v_4}$};
				\draw[lred] (xe4v3) -- (zv3v1v4);
				\draw[densely dotted] (zv3v1v4) -- (yv3v1v4);
				
				\node[xnode] (zv3v2t2) at (11,0) {};
				\node[above] at (zv3v2t2) {$z^{v_3}_{v_2,t_2}$};
				\node[xnode] (yv3v2t2) at (12,0) {};
				\node[above] at (yv3v2t2) {$y^{v_3}_{v_2,t_2}$};
				\draw[blue] (xe5v3) -- (zv3v2t2);
				\draw[dashed] (zv3v2t2) -- (yv3v2t2);
				
				\node[xnode] (v3) at (13,1) {};
				\node[below] at (v3) {$v_3$};
				\node[xnode] (v'3) at (13,2) {};
				\node[above] at (v'3) {$v'_3$};
				\draw[dashed] (v3) -- (v'3);
				\draw[lred] (yv3v1v4) -- (v3);
				\draw[blue] (yv3v2t2) -- (v3);
				
				\node[xnode] (yv3v4v1) at (14,2) {};
				\node[above] at (yv3v4v1) {$y^{v_3}_{v_4,v_1}$};
				\node[xnode] (zv3v4v1) at (15,2) {};
				\node[above] at (zv3v4v1) {$z^{v_3}_{v_4,v_1}$};
				\draw[lred] (v3) -- (yv3v4v1);
				\draw[densely dotted] (yv3v4v1) -- (zv3v4v1);
				
				\node[xnode] (xe6v3) at (16,2) {};
				\node[above] at (xe6v3) {$x^{e_6}_{v_3}$};
				\node[xnode] (xe6v4) at (17,2) {};
				\node[above] at (xe6v4) {$x^{e_6}_{v_4}$};
				\draw[lred] (zv3v4v1) -- (xe6v3);
				\draw[densely dotted] (xe6v3) -- (xe6v4);
				
				\node[xnode] (v5) at (11,-2) {};
				\node[below] at (v5) {$v_5$};
				\node[xnode] (v'5) at (11,-1) {};
				\node[above] at (v'5) {$v'_5$};
				\draw[densely dotted] (v5) -- (v'5);
				
				\node[xnode] (zv4v1t1) at (18,4.5) {};
				\node[above] at (zv4v1t1) {$z^{v_4}_{v_1,t_1}$};
				\node[xnode] (yv4v1t1) at (19,4.5) {};
				\node[above] at (yv4v1t1) {$y^{v_4}_{v_1,t_1}$};
				\node[xnode] (zv4v1v3) at (18,3.5) {};
				\node[above] at (zv4v1v3) {$z^{v_4}_{v_1,v_3}$};
				\node[xnode] (yv4v1v3) at (19,3.5) {};
				\node[above] at (yv4v1v3) {$y^{v_4}_{v_1,v_3}$};
				\node[xnode] (zv4v3v1) at (18,2.5) {};
				\node[above] at (zv4v3v1) {$z^{v_4}_{v_3,v_1}$};
				\node[xnode] (yv4v3v1) at (19,2.5) {};
				\node[above] at (yv4v3v1) {$y^{v_4}_{v_3,v_1}$};
				\node[xnode] (zv4v3t1) at (18,1.5) {};
				\node[above] at (zv4v3t1) {$z^{v_4}_{v_3,t_1}$};
				\node[xnode] (yv4v3t1) at (19,1.5) {};
				\node[above] at (yv4v3t1) {$y^{v_4}_{v_3,t_1}$};
				
				\draw[red] (xe3v4) -- (zv4v1t1) {};
				\draw[lred] (xe3v4) -- (zv4v1v3) {};
				\draw[lred] (xe6v4) -- (zv4v3v1) {};
				\draw[lred] (xe6v4) -- (zv4v3t1) {};
				\draw[dashed] (zv4v1t1) -- (yv4v1t1);
				\draw[densely dotted] (zv4v1v3) -- (yv4v1v3);
				\draw[densely dotted] (zv4v3v1) -- (yv4v3v1);
				\draw[densely dotted] (zv4v3t1) -- (yv4v3t1);
				
				\node[xnode] (v4) at (20,3) {};
				\node[below] at (v4) {$v_4$};
				\node[xnode] (v'4) at (20,4) {};
				\node[above] at (v'4) {$v'_4$};
				\draw[dashed] (v4) -- (v'4);
				
				\draw[red] (yv4v1t1) -- (v4);
				\draw[lred] (yv4v1v3) -- (v4);
				\draw[lred] (yv4v3v1) -- (v4);
				\draw[lred] (yv4v3t1) -- (v4);
				
				\node[xnode] (yv4t1v1) at (21,3.5) {};
				\node[above] at (yv4t1v1) {$y^{v_4}_{t_1,v_1}$};
				\node[xnode] (zv4t1v1) at (22,3.5) {};
				\node[above] at (zv4t1v1) {$z^{v_4}_{t_1,v_1}$};
				\node[xnode] (yv4t1v3) at (21,2.5) {};
				\node[above] at (yv4t1v3) {$y^{v_4}_{t_1,v_3}$};
				\node[xnode] (zv4t1v3) at (22,2.5) {};
				\node[above] at (zv4t1v3) {$z^{v_4}_{t_1,v_3}$};
				
				\draw[red] (v4) -- (yv4t1v1);
				\draw[lred] (v4) -- (yv4t1v3);
				\draw[dashed] (yv4t1v1) -- (zv4t1v1);
				\draw[densely dotted] (yv4t1v3) -- (zv4t1v3);
				
				\node[xnode] (xe8v4) at (23,3) {};
				\node[above] at (xe8v4) {$x^{e_8}_{v_4}$};
				\node[xnode] (xe8t1) at (24,3) {};
				\node[above] at (xe8t1) {$x^{e_8}_{t_1}$};
				
				\draw[red] (zv4t1v1) -- (xe8v4);
				\draw[lred] (zv4t1v3) -- (xe8v4);
				\draw[dashed] (xe8v4) -- (xe8t1);
				
				\node[xnode] (t1) at (24,2) {};
				\node[left] at (t1) {$t_1$};
				\node[xnode] (t2) at (24,1) {};
				\node[left] at (t2) {$t_2$};
				
				\draw[red] (xe8t1) -- (t1);
				\draw[dashed] (t1) -- (t2);
				
				\node[xnode] (yv3t2v2) at (14,0) {};
				\node[above] at (yv3t2v2) {$y^{v_3}_{t_2,v_2}$};
				\node[xnode] (zv3t2v2) at (15,0) {};
				\node[above] at (zv3t2v2) {$z^{v_3}_{t_2,v_2}$};
				\draw[blue] (v3) -- (yv3t2v2);
				\draw[dashed] (yv3t2v2) -- (zv3t2v2);
				
				\node[xnode] (xe7v3) at (16,0) {};
				\node[above] at (xe7v3) {$x^{e_7}_{v_3}$};
				\node[xnode] (xe7t2) at (17,0) {};
				\node[above] at (xe7t2) {$x^{e_7}_{t_2}$};
				
				\draw[blue] (zv3t2v2) -- (xe7v3);
				\draw[dashed] (xe7v3) -- (xe7t2);
				\draw[blue] (xe7t2) -- (t2);
			\end{tikzpicture}
			\caption{The instance of \cec{} resulting from applying the reduction in the proof of \cref{thm:paraNPh-above-matching} to the instance of \textsc{TDMP} in (a).
			A stable edge set $F$ of size $M(G) + 1$ is indicated
			Dashed and dotted edges have color $c^m$, where the former are not part of $F$ and the latter are.
			Edges drawn in red (blue) result from red (blue) edges in the instance of \textsc{TDMP}.
			Not all red (blue) edges actually have the same color in the instance of \cec{}.
			An edge is light red (light blue) if it is not in $F$ and dark red (dark blue) if it is.}
		\end{subfigure}
		\caption{Illustration of the reduction in the proof of \cref{thm:paraNPh-above-matching}.}
		\label{fig:paraNP-above-matching}
	\end{figure}
	Wu~\cite{Wu2012} showed that the problem \textsc{Two Disjoint Monochromatic Paths} (\textsc{TDMP}) is \NP-hard.
	It is defined as follows.
	The input consists of a graph $G=(V,E)$ with a two-edge coloring $\ell \colon E \to \{1,2\}$ and terminal vertices $s_1,t_1,s_2,t_2\in V$.
	The task is to find an $s_1$-$t_1$~path~$P_1$ that uses only edges of color $1$~and an $s_2$-$t_2$ path $P_2$ that uses only edges of color~$2$ such that no vertex is visited by both $P_1$ and $P_2$.
	A closer examination of Wu's proof reveals that \textsc{TDMP} is \NP-hard even if every vertex is incident to at most six edges with color~$1$ and at most three edges with color~$2$.
	By a very simple modification, we may additionally assume that all terminal vertices have degree~$1$, where the sole edge incident to $s_1$ has color~$1$ and the sole edge incident to $s_2$ has color~$2$.
	We will give a reduction from \textsc{TDMP} with these restrictions to \cec{} with $k-M(G) = 1$ and $\abs{C}\leq 19$.
	This reduction is illustrated in \Cref{fig:paraNP-above-matching}.
	
	Let $(G=(V,E),\ell,s_1,t_1,s_2,t_2)$ be an instance of \textsc{TDMP} with the aforementioned restrictions.
	We construct an instance $(G'=(V',E'),\ell',k)$ of \cecS{} as follows.
	Let $c^m \in C$ denote a fixed color, which will be used for all matching edges (except for two edges $\{ s_1, s_1' \}$ and $\{ s_2, s_2' \}$) in $G'$.
	We also partition the colors in $C\setminus \{c^m\}$ into $C_1$ and $C_2$ with $\abs{C_1} = 15$ and $\abs{C_2} = 3$.
	We add all vertices in $G$ to~$G'$.
	For every edge $e=\{u,v\} \in E$, we add vertices $x^e_u$ and $x^e_v$ to $G'$ and connect the two by an edge with color $c^m$.
	For any vertex $v \in V \setminus \{s_1,t_1,s_2,t_2\}$, we do the following.
	First, we add a vertex $v'$, which we connect to $v$ with an edge of color~$c^m$.
	Let $w_1,\ldots,w_a$ denote the vertices connected to $v$ by an edge with color~$1$ and $x_1,\ldots,x_b$ those connected by an edge with color~$2$.
	For any $\{i,j\} \in \binom{[a]}{2}$, choose a color $c^v_{\{w_i,w_j\}} \in C_1$ and for any $\{i',j'\} \in \binom{[a]}{2}$ a color~$d^v_{\{x_{i'},x_{j'}\}} \in C_2$.
	All $c^v_{\{w_i,w_j\}}$ and $d^v_{\{x_{i'},x_{j'}\}}$ for a fixed $v$ must be pairwise distinct.
	Note that because $v$ has at most six incident color-$1$ edges and at most three incident color-$2$ edges, this requires at most $\binom{6}{2} = 15$ colors in $C_1$ and $\binom{3}{2} = 3$ colors in $C_2$.
	For each $i \in [a]$ and $j \in [a]\setminus\{i\}$, we do the following: 
	\begin{itemize}
		\item 
		Add a vertex $y_{w_i,w_j}^v$ and we connect it with $v$ by an edge with color $c^v_{\{w_i,w_j\}}$.
		So $v$ has $a(a - 1)$ incident edges (in addition to $\{ v, v' \}$) and each color $c^v_{\{w_i,w_j\}}$ appears exactly twice.
		\item
		Add a vertex $z_{w_i,w_j}^v$ and we connect it with $y_{w_i,w_j}$ by an edge with color $c^m$ and with $x^{\{v,w_i\}}_v$ by an edge with color $c^v_{\{w_i,w_j\}}$.
		Note that every vertex $x_v^{v, w_i}$ has $a - 1$ incident edges (in addition to $\{ x_v^{v, w_i}, x_{w_i}^{v, w_i} \}$) and each color $c^v_{\{w_i,w_j\}}$ appears exactly once.
	\end{itemize}
	For each~$i'\in [b]$, we repeat this construction, using the colors $d^v_{\{x_{i'},x_{j'}\}}$.
	The intuitive idea is that when $P_1$ goes through $w_i$, $v$, and then $w_j$, the vertices $x_v^{\{v, w_i\}}$, $x_v^{\{v, w_j\}}$, and $v$ will be colored with $c^v_{\{w_i,w_j\}}$.
	On the other hand, if neither $P_1$ nor $P_2$ goes through $v$, then $v$ will be colored with $c^m$.
	Finally, we consider the terminal vertices~$s_1,t_1,s_2,t_2$.
	For $s_i$, $i\in\{1,2\}$, we do the following.
	We add a vertex $s'_i$ and connect $s_i$ to $s'_i$ by an edge with an arbitrary color~$c\in C_i$.
	We also connect $s_i$ to $x^e_{s_i}$, where $e$ is the sole edge incident to~$s_i$, by an edge with color $c$.
	For $t_i$, $i\in\{1,2\}$, we do the following.
	We connect $t_i$ to $x^e_{t_i}$, where $e$ is the edge incident to~$t_i$ by an edge with an arbitrary color~$c\in C_i$ and to $t_{3-i}$ by an edge of color $c^m$.
	
	Let $M \coloneqq \{e \in E' \mid \ell'(e) = c^m\} \cup \{\{s_1,s'_1\},\{s_2,s'_2\}\}$.
	Note that $M$ is a perfect matching in~$G'$.
	We set $k\coloneqq \abs{M} + 1$.
	Clearly, the output instance can be computed in polynomial time.
	It remains to show that $(G,\ell,s_1,t_1,s_2,t_2)$ is a \yes-instance for \textsc{TDMP} if and only if $(G',\ell',k)$ is a \yes-instance for \cecS{}.
	
	\RD{} Let $s_1=v_1,v_2,\ldots,v_p=t_1$ and $s_2=\tv_1,\tv_2,\ldots,\tv_q=t_2$ be monochromatic vertex-disjoint paths of color~$1$ and $2$, respectively.
	Define a stable edge set in $G'$ as follows:
	{\allowdisplaybreaks
	\begin{align*}
                F\coloneqq & \{ \{s_1,s'_1\},\{s_2,s'_2\}\} \cup \{\{v,v'\} \mid v \in V, v\notin \{v_1,\ldots,v_p,\tv_1,\ldots,\tv_q\}\}\\
                &\cup \{\{y^v_{u,w},z^v_{u,w}\} \mid \{u,v\},\{v,w\} \in E, \ell(\{u,v\}) = \ell(\{v,w\}), \\
                &\quad \quad   \nexists i\colon (v_{i-2} = u \wedge v_i = v \wedge v_{i+1}=w) \vee (v'_{i-1} = u \wedge \tv_i = v \wedge  \tv_{i+1}=w)\} \\             
                &\cup \{\{x^e_u,x^e_v\} \mid e \in E, \nexists i\colon (v_i = u \wedge v_{i+1}=v) \vee (\tv_i = u \wedge \tv_{i+1}=v)\} \\
                &\cup \{ \{x^{\{v_{i-1},v_i\}}_{v_i},z^{v_i}_{v_{i-1},v_{i+1}}\}, \{y^{v_i}_{v_{i-1},v_{i+1}},v_i\}, \{v_i,y^{v_i}_{v_{i+1},v_{i-1}},v_i\},\\
                &\quad \,\{x^{\{v_{i},v_{i+1}\}}_{v_i},z^{v_i}_{v_{i+1},v_{i-1}}\} \mid i \in \{2,\ldots,p-1\}\}\\
                &\cup \{ \{x^{\{\tv_{i-1},\tv_i\}}_{\tv_i},z^{\tv_i}_{\tv_{i-1},\tv_{i+1}}\}, \{y^{\tv_i}_{\tv_{i-1},\tv_{i+1}},\tv_i\}, \{\tv_i,y^{\tv_i}_{\tv_{i+1},\tv_{i-1}},\tv_i\},\\
                &\quad \,\{x^{\{\tv_{i},\tv_{i+1}\}}_{\tv_i},z^{\tv_i}_{\tv_{i+1},\tv_{i-1}}\} \mid i \in \{2,\ldots,q-1\}\}.
	\end{align*}
	}
	This edge set is stable.
	Moreover, $\deg_F(e) = 2$ for all $e\in M\setminus \{ \{s_1,s'_1\},\{s_2,s'_2\}\}$, while $\deg_F(\{s_1,s'_1\}) = \deg_F(\{s_1,s'_1\}) = 3$.
	By \cref{obs:matching-counting}, it follows that $\abs{F} = \abs{M} + 1 = k$.
	
	\LD{}
	Suppose that $F$ is a stable edge set in $G'$ with $\abs{F} = k = \abs{M} + 1$.
	No edge $e \in M\setminus \{ \{s_1,s'_1\},\{s_2,s'_2\}\}$ can have $\deg_F(e) \geq 2$.
	Moreover, the $\deg_F(\{s_1,s'_1\}), \deg_F(\{s_2,s'_2\}) \leq 3$.
	By \cref{obs:matching-counting}, it follows that $\deg_F(e) = 2$ for all $e \in M\setminus \{ \{s_1,s'_1\},\{s_2,s'_2\}\}$ and $\deg_F(\{s_1,s'_1\}) = \deg_F(\{s_1,s'_1\}) = 3$.
	A \emph{graph with loops} is a hypergraph $H=(W,F)$ such that $\abs{e} \in \{1,2\}$ for all $e\in F$.
	Consider the auxiliary graph with loops $H=(M,E'')$ where $E''$ contains a loop on the vertex $e \in M$ if $e\in F$ and an edge between $e,e'\in M$, $e\neq e'$, if $F$ contains an edge $f$ with $f\cap e \neq \emptyset \neq f\cap e'$.
	In $H$, the degree of any $e \in M\setminus \{ \{s_1,s'_1\},\{s_2,s'_2\}\}$ is $2$, while $\{s_1,s'_1\},\{s_2,s'_2\}$ have degree $3$.
	It follows that, in $H$, the vertices $\{s_1,s'_1\},\{s_2,s'_2\}$ have a loop and an additional neighbor each, while all other vertices either have a loop or two neighbors distinct from themselves.
	This means that $H$ contains a path between $\{s_1,s'_1\}$ and $\{s_2,s'_2\}$.
	
	Let $\{s_1,s_1'\} = e_1,e_2,\ldots,e_p=\{s_2,s'_2\}$ be the aforementioned path in $H$.
	In the following, we will show that this path essentially contains two disjoint monochromatic paths, from $s_1$ to~$t_1$ and from $s_2$ to~$t_2$, respectively.	
	For any $i\in\{1,\ldots,p\}$, if $e_i \neq \{t_1,t_2\}$, then the colors of the edges in $F$ witnessing that $\deg_F(e_i) \geq 2$ are all from $C_1$ or all from $C_2$.
	First, notice that they cannot have color $c^m$ because any such edges would form a loop in $H$ and therefore cannot be on this path.	
	If $e_i= \{v,v'\}$, then this claim is true because
	$\deg_{G'}(v') = 1$ and therefore both of the edges must be incident to $v$ and, therefore have the same color.
	In all other cases, no vertex in $e_i$ is incident to both an edge with a color in $C_1$ and an edge with a color $C_2$.
	Let $\ell(e_i) \coloneqq j \in \{1,2\}$ where the edges witnessing that $\deg(e_i) \geq 2$ are all from $C_j$.
	By a similar argument, we can show that $\ell(e_i) = \ell(e_{i+1})$ if $e_i \neq \{t_1,t_2\} \neq e_{i+1}$.
	Since $\ell(e_1) = 1$ and $\ell(e_p) = 2$, it follows that $\{t_1,t_2\}$ must be visited by the path, say $e_q = \{t_1,t_2\}$.
	
	In $H$, $\{s_i,s'_i\}$ is only adjacent to a single vertex $\{x^e_{s_i},x^e_{v}\}$, where $v$ is $s_i$'s sole neighbor in $G$.
	Any vertex in $H$ of the type $\{x^e_{u},x^e_{v}\}$, with $e=\{u,v\}$, is only adjacent to vertices of type $\{y^u_{v,w}, z^u_{v,w}\}$, $\{s_1, s_2\}$, or $\{t_1,t_2\}$.
	A vertex in $H$ of type~$\{y^u_{v,w}, z^u_{v,w}\}$ is adjacent to exactly one vertex of the type $\{x^e_{u},x^e_{v}\}$ and one of the type $\{u,u'\}$.
	Finally, in~$H$, vertices of the type~$\{v,v'\}$ are only adjacent to vertices of the type $\{y^v_{u,w}, z^v_{u,w}\}$.
	It follows that, for any $i \in \{2,\ldots,q-1\}$,
	\begin{itemize}
		\item if $i \bmod 4 = 2$, then $e_i$ is of the type $\{x^e_{u},x^e_{v}\}$,
		\item if $ i \bmod 4 \in \{1,3\}$, then $e_i$ is of the type $\{y^u_{v,w}, z^u_{v,w}\}$, and
		\item if $i \bmod 4 = 0$, then $e_i$ is of the type $\{u,u'\}$.
	\end{itemize}
	Then, $s_1=v_1,v_2,\ldots,v_{\frac{q+5}{4}}=t_1$ where for any $i\in\{2,\ldots, \frac{q+5}{4}\}$, we have $e_{4i-1} = \{v_i, v_i'\}$, is a path from $s_1$ to $t_1$ that uses only edges of color~$1$.
	Similarly, one can get a path from $s_2$ to $t_2$ that is disjoint from the first path and only uses edges with color~$2$ by considering $e_q,\ldots,e_p$.
	This proves that the instance of \textsc{TDMP} is a \yes-instance.
\end{proof}

Finally, we will show that the \XP{} algorithm for $k-I(G)$ most likely also cannot be improved to an \FPT{} algorithm.
	
\begin{theorem}
	\label{thm:w1h-above-IM}
	\cec{} is \W{1}-hard with respect to $k - I(G)$.
\end{theorem}
\begin{proof}
	We will prove \W{1}-hardness by reduction from \textsc{Multicolored Clique}.
	In this problem, the input consists of a graph $G=(V,E)$ where the vertex set is partitioned into $s$ independent sets $V=V_1\uplus \ldots \uplus V_s$ and the task is to find a clique which contains exactly one vertex from each $V_i$.
	This problem is \W{1}-hard when parameterized by $s$~\cite{Fellows2009,Pietrzak2003}.
	
	Given an instance $(G=(V=V_1\uplus\ldots\uplus V_s,E),s)$ of \textsc{Multicolored Clique}, we 
construct an instance $(G'=(V',E'),\ell,k)$ of \cec{} as follows.
	Let $C = \{c^m,c^e\} \cup \{c_v \mid v \in V\}$, that is, for each vertex in $G$ there is a color~$c_v$ in~$G'$ in addition to the colors $c^m$ and $c^e$.
	All matching edges will have color $c^m$.
	For each $i \in [s]$, the graph $G'$ contains a vertex $u_i$ and for each~$j \in [s] \setminus \{i\}$ and $v \in V_i$ it contains the vertices $w^{i\to j}_v$ and $x^{i\to j}_v$.
	Finally, for each~$e \in E$, we add a vertex~$y_e$ to~$G'$.
	The edges of $G'$ are as follows.
	For each $i \in [s]$,  $v \in V_i$, and $j\in [s]\setminus \{j\}$, there is an edge from $u_i$ to $w^{i \to j}_v$ with color $c_v$ and an edge from $w^{i\to j}_v$ to $x^{i\to j}_v$ with color $c^m$.
	For each edge $e = \{v,v'\} \in E$ with $v \in V_i$ and $v' \in V_j$, there are edges with color $c^e$ from~$x^{i\to j}_v$ and $x^{j\to i}_{v'}$ to $y_e$.
	We let $k \coloneqq (s-1) \cdot (\abs{V} + s)$.
	Note that all edges with color $c^m$ form an induced matching and that the number of such edges is $(s-1)\cdot \abs{V}$.
	Therefore, $k - I(G) \leq s(s-1)$.
	The intended meaning of our construction is that when a vertex $v \in V_i$ is part of a clique $X$ of size $k$, the vertex $u_i$ will be colored by $c_v$ and the edges $\{ u_i, w_v^{i \to j} \}$ along with $\{ w_v^{i \to j}, y_e \}$ for edges $e$ incident to $v$ in $G[C]$ will be included into a stable set $F$.
	Essentially, the edges incident to $y_e$ for the edges $e$ having both endpoints in $X$ will account for the increase in the stable set size (when compared to~$M$).
	Clearly, this construction can be computed in polynomial time.
	It remains to show that $(G,s)$ is \yes-instance for \textsc{Multicolored Clique} if and only $(G',\ell,k)$ is a \yes-instance for \cec{}.
	
	\RD{} Suppose that $X = \{v_1,\ldots,v_s\}$ with $v_i \in V_i$ is a clique in $G$.
	Then,
	\begin{align*}
		F \coloneqq& \{ \{u_i,w^{i\to j}_{v_i}\} \mid i \in [s], j \in [s] \setminus \{i\} \}\\
		& \cup \{ \{w_v^{i\to j},x_v^{i \to j}\} \mid i \in [s], j \in [s] \setminus \{i\}, v \in V_i \setminus \{v_i\}\}\\ & \cup \{ \{y_e,x^{i \to j}_{v_i}\} \mid e = \{v_i,v_j\} \in E \}
	\end{align*}
	is a stable edge set in $G'$.
	Moreover, 
	\begin{align*}
		\abs{F} & = s \cdot (s-1) + \sum_{i\in [s]} (\abs{V_i} -1)(s-1) + s\cdot(s-1) \\
		&= s\cdot(s-1) + (s-1)(\abs{V} -s) + s\cdot(s-1) \\
		& = (s-1)(\abs{V} + s).
	\end{align*}
	
	\LD{} Suppose that $F$, $\abs{F} = k$, is a stable edge in $G'$ and that $f_F \colon V' \to C$ is the corresponding vertex coloring.
	We may assume that of all stable edge sets of size $k$, $F$ contains a maximum number of edges in the induced matching (which consists of all edges with $c^m$).
	This implies that, if for any $e = \{v,v'\}\in E$ with $v \in V_i$ and $v' \in V_j$ the edge $\{y_e, x^{i\to j}_v\}$ is contained in $F$, then so are the edges $\{u_i,w_v^{i\to j}\}$, $\{y_e,x_{v'}^{j\to i}\}$, and $\{u_j,w_{v'}^{j\to i}\}$.
	Otherwise, $\{y_e, x^{i\to j}_v\}$ could be replaced by $\{x^{i\to j}_v, w^{i\to j}_v\}$ without decreasing the size of $F$.
	Similarly, if $\{u_i,w^{i \to j}_v\} \in F$, then there is an edge $e =\{v,v'\} \in E$, $v\in V_i,v' \in V_j$, such that $\{x^{i \to j}_v,y_e\}, \{x^{j \to i}_v,y_e\}, \{u_j,w^{j \to i}_{v'}\} \in F$.
	Let $X \coloneqq \{v \in V \mid \exists e \in E \colon \{y_e,x^{i\to j}_v\} \in F, v \in e\}$.
	Next, let $Y \coloneqq \{ \{i,j\} \in \binom{[s]}{2} \mid \exists e = \{v,v'\} \in E \colon \{x^{i\to j}_v,y_e  \},\{x^{j\to i}_{v'},y_e \}\in F \}$.
	Note that $F$ contains $2\abs{Y}$ edges of the type $\{u_i,w_{v}^{i \to j}\}$ and another $2 \abs{Y}$ edges of the type $\{x^{i\to j}_v,y_e  \}$.
	It contains $\abs{V}\cdot(s-1) - 2\abs{Y}$ edges of the type $\{w_v^{i\to j},x_v^{i \to j}\}$.
	Hence,
	\begin{align*}
		(s-1)(\abs{V} +s) = k = 4\abs{Y} + \abs{V}(s-1) - 2\abs{Y} = \abs{V}(s-1) - 2\abs{Y},
	\end{align*}
	implying that $\abs{Y} = \binom{s}{2}$.
	Hence, for any $\{i,j\} \in \binom{[s]}{2}$, there are $v\in V_i \cap X$ and $v' \in V_j \cap X$ such that the edges $\{x^{i\to j}_v,y_e\}$ and $\{x^{j\to i}_{v'},y_e\}$ are contained in $F$, implying that $v$ and $v'$ are adjacent in $G$.	
	Therefore, $X$ is a clique in $G$.
\end{proof}

\section{Structural parameters}
\label{sec:struct}
In the following, we will classify the parameterized complexity of \cec{} with respect to structural graph parameters.
Our results are summarized in \Cref{fig:hier}.
The problem is \W{1}-hard with respect to vertex cover number, which rules out FPT algorithms for many graph parameters.
Additionally, the problem is also \W{1}-hard with respect to tree-cut width.
On the positive side, it is \fpt{} for the slim tree-cut width, a parameter that was very recently introduced for the express purpose of dealing with problems that are hard for tree-cut width and treewidth~\cite{Ganian2022}.
Hence, \cec{} can be added to a list of problems compiled by Ganian and Korchemma~\cite{Ganian2022} that are hard for tree-cut width, but \fpt{} for slim tree-cut width.

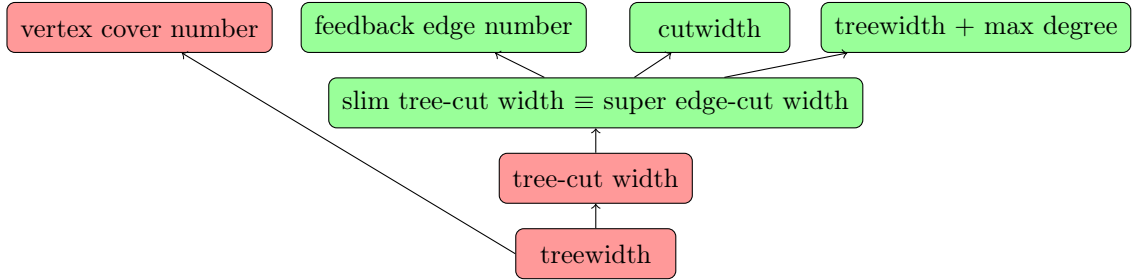
\begin{figure}
	\centering
	\begin{tikzpicture}[yscale=.5]
	\tikzset{
		param/.style={draw, fill=white, rectangle, rounded corners=3, font=\small, minimum width=5.5em, minimum height=4ex},
	}
	
	\node[param, inner sep=5pt,inner ysep=+0pt,fill=red!40] at (-2, 0) (tw) {treewidth};
	\node[param, inner sep=5pt,inner ysep=+0pt,fill=red!40] at (-2, 2) (tcw) {tree-cut width};
	\node[param, inner sep=5pt,inner ysep=+0pt,fill=green!40] at (-2, 4) (stcw) {slim tree-cut width $\equiv$ super edge-cut width};
	\node[param, inner sep=5pt,inner ysep=+0pt,fill=green!40] at (-4, 6) (fen) {feedback edge number};
	\node[param, inner sep=5pt,inner ysep=+0pt,fill=green!40] at (-0.5, 6) (cw) {cutwidth};
	\node[param, inner sep=5pt,inner ysep=+0pt,fill=green!40] at (3, 6) (twd) {treewidth $+$ max degree};
	\node[param, inner sep=5pt,inner ysep=+0pt,fill=red!40] at (-8, 6) (vc) {vertex cover number};
	\draw[->] (tw) -- (tcw);
	\draw[->] (tcw) -- (stcw);
	\draw[->] (stcw) -- (fen);
	\draw[->] (stcw) -- (twd);
	\draw[->] (stcw) -- (cw);
	\draw[->] (tw.west) -- (vc);
	
	\end{tikzpicture}
	\caption{A Hasse diagram of graph parameters relevant to this section.
		\cec{} is \fpt{} with respect to parameters highlighted in green and \W{1}-hard with respect to those in red.
		There is an arrow from a parameter $p$ to a parameter $q$ if there is a function $f$ such that $f(p(G))\ge q(G)$ holds for every graph $G$.}
	\label{fig:hier}
\end{figure}

We start with a more general observation.
A graphs class $\calC$ is \emph{monotone} if $G \in \calC$ and $H \subseteq G$ implies that $H \in \calC$.
For a graph $G=(V,E)$ and a monotone graph class $\calC$, let $d^e_\calC(G) \coloneqq \min_{E' \subseteq E, G-E' \in \calC} \abs{E'}$ denote the \emph{edge deletion distance} of $G$ to $\calC$.

\begin{proposition}
	\label{prop:distance-to-poly}
	Let $\calC$ be a monotone class of graphs such that \cec{} is polynomial-time solvable on $\calC$.
	Then, \cec{} is \fpt{} with respect to~$d^e_\calC$ if a minimum edge deletion set to $\calC$ is given as part of the input.
\end{proposition}
\begin{proof}
	Let $E' \subseteq E$ be an edge set of size at most $d_\calC^e(G)$ with $G-E' \in \calC$.
	The general idea behind the algorithm is that we test every subset $F\subseteq E'$ of $E'$ and assume that the edges in~$F$ are stable while those in $E'\setminus F$ are not.
	We then use the polynomial-time algorithm for~$\calC$ to find a maximum edge set in $E\setminus E'$ that can be added to $F$.
	
	For each $F\subseteq E'$, we either return that this choice of $F$ does not yield a solution or we compute an equivalent instance $(G_{F}=(V,E_F),\ell_F,k_F)$ where $G_F$ is a subgraph of $G-E'$.
	Note that, because $\calC$ is monotone, this implies that $G_F \in \calC$.
	We obtain $G_F$ by first deleting every edge in $E' \setminus F$.
	For each edge $e \in F$, we check whether there is an $e' \in F$ with $e \cap e' \neq \emptyset$ and $\ell(e) \neq \ell(e')$.
	If such an edge exists, then we return that this choice of $F$ does not yield a solution.
	Otherwise, for each vertex $v \in V$ that is incident to an edge $e \in F$, we remove the edges $\{e' \in E \setminus E' \mid \ell(e') \neq \ell(e), v \in e'\}$.
	The edge set obtained by these deletions is $E_F$, $\ell_F$ is~$\ell$ restricted to $E_F$, and $k_F \coloneqq k - \abs{F}$.
	
	We accept the input instance if the polynomial-time algorithm on $\calC$ returns yes for at least one $F\subseteq E'$.
	This way we obtain a running time of $\bigO(2^{d_\calC^e(G)} \cdot n^c)$ where $\bigO(n^c)$ is the running time of the algorithm on $\calC$.
\end{proof}

\cref{prop:distance-to-poly} directly implies that \cec{} is \fpt{} with respect to feedback edge number (\cecS{} is polynomial-time solvable on forests) and with respect to the number of edges that do not have one of the two most frequent colors (\cecS{} is polynomial-time solvable for $|C|=2$).
For vertex deletion distance, a similar statement is not true, since the problem is \W{1}-hard with respect to vertex cover number as we will show in \cref{thm:w1-vc}.

It is also easy to prove that \cec{} is \FPT{} with respect to the joint parameterization by the treewidth of the input graph and the maximum number of colors incident to any vertex.
This can be shown by a standard dynamic program on the tree decomposition by iterating over all colorings of each bag.
Since the maximum number of colors incident to any vertex is at most the maximum degree, it follows that \cecS{} is also \FPT{} with respect to treewidth plus maximum degree.

We will now consider a parameter that is smaller than feedback edge number and treewidth plus maximum degree.
The graph parameter \emph{slim tree-cut width} was introduced by Ganian and Korchemna~\cite{Ganian2022} as a way of dealing with problems that remain \W{1}-hard with respect to treewidth and tree-cut width.
We show that \cec{} is \fpt{} with respect to slim tree-cut width, by considering the asymptotically equivalent parameter \emph{super edge-cut width}.
Let $G=(V,E)$ be a graph and $T=(V,E')$ a tree on the same vertex set.
For any $v \in V$, the \emph{local feedback edge set} at $v$ is \[\lfe (G,T,v) \coloneqq \abs{\{ \{u,w\} \in E \setminus E' \mid \text{$v$ is on the unique $u$-$w$-path in $T$}\} },\]
that is, $\lfe(G,T,v)$ counts the number of edges that are not in $T$ such that the unique path in~$T$ which connects the endpoints visits $v$.
Note that possibly~$u=v$ or~$w=v$.
The \emph{local feedback edge number} of $(G,T)$ is $\lfe(G,T) \coloneqq \max_{v\in V} \lfe(G,T,v)$.
The \emph{super edge-cut width}~\cite{Ganian2022} of $G$ is \[\secw(G) \coloneqq 1 + \min_{T \text{ is a tree on $V$}} \lfe(G,T).\]
Ganian and Korchemna~\cite{Ganian2022} showed that there is an algorithm with running time $\bigO(f(k) \cdot \abs{G}^{\bigO(1)})$ that given $k \in \N$ and a graph $G=(V,E)$ either outputs a tree $T$ on $V$ such that $\lfe(G,T) \leq \bigO(k^6)$ or correctly determines that $\secw(G) > k$.
We will not formally define slim tree-cut width, but since it is asymptotically equivalent to super edge-cut width, it is sufficient to consider the latter.

\begin{theorem}
	\cec{} is \fpt{} with respect to super edge-cut width and slim tree-cut width.
\end{theorem}
\begin{proof}
	\newcommand{\tD}{\tilde{D}}
	\newcommand{\ch}{\mathrm{ch}}
	Let $(G=(V,E),\ell,k)$ be an instance of \cec{}.
	We compute a tree $T=(V,E')$ with $\lfe(G,T) \leq \bigO(\secw(G)^6)$, which we use in a dynamic programming algorithm that solves this instance in time $\bigO(2^{\lfe(G,T)} \cdot \abs{\colors}^2 n^2)$.
	Let $r\in V$ be an arbitrary vertex that we designate as the root.
	For each $v\in V$, let $T(v) \subseteq V$ be the set all descendants of $v$, not including $v$, and let $T[v] \coloneqq T(v) \cup \{v\}$.
	We let
	\begin{align*}
	C_1(v) & \coloneqq \{ \{u,w\} \in E \setminus E' \mid u \in T(v), w \in V \setminus T[v]\},\\
	C_2(v) & \coloneqq \{ \{u,w\} \in E \setminus E' \mid u \in T[u'], w \in T[w'], u',w' \text{ are distinct children of } v\},\\
	C_3(v) & \coloneqq \{ \{v,w\} \in E \setminus E' \mid w \in T(v)\}, \\
	C_4(v) & \coloneqq \{ \{u,v\} \in E \setminus E' \mid u \in V \setminus T[v]\}.
	\end{align*}
	Note that $\abs{C_1(v)} + \abs{C_2(v)} + \abs{C_3(v)} + \abs{C_4(v)} \leq \lfe(G,T)$.
	Each edge in $C_1(v) \cup C_4(v)$ has exactly one endpoint in $T[v]$, whereas the edges in $C_2(v) \cup C_3(v)$ are fully contained in $T[v]$.
	Finally, we denote the set of children of $v$ by $\ch(v)$.
	For any color $c$ we use $\ch_c(v)$ to refer to the set of all children of $v$ whose edge to $v$ is present in $G$ and has color $c$ (observe that $T$ is not necessarily a subgraph of $G$, so $v$ may have children in $T$ that are not adjacent to $v$ in $G$).

	Our algorithm is a dynamic program that computes a table $D$.
	For each $v\in V$, $c \in C$, and each stable set $S \subseteq C_1(v) \cup C_4(v)$, the table entry $D[v,c,S]$ contains the size of a largest stable edge set $F$ in $T[v]$ that is compatible with $S$, meaning that~$F \cup S$ must also be stable, and where the vertex $v$ receives color $c$, that is any edge in $F \cup S$ incident to $v$ must have color~$c$. 
	Clearly, $D[v,c,S] = 0$ for all leaves $v$ of $T$, as $T[v]$ does not contain any edges if $v$ is a leaf.
	
	Now suppose that $v$ is an inner node of $T$.
	We start by giving the recursion formula for $D[v,c,S]$.
 	If $S \cap C_4(v) \neq \emptyset$, then there must be a color $c\in C$ such that $\ell(e) = c$ for all $e \in S \cap C_4(v)$.
 	This is because all these edges are incident to~$v$.
 	In this case, the color of $v$ is determined to be $c$, so  	$D[v,c',S] = 0$ for all $c' \in \colors \setminus \{c\}$.
 	If $S \cap C_4(v) = \emptyset$ or if $c$ is the color of all edges in $S \cap C_4(v)$, then the formula is as follows (we use $E_c$ for $c \in \colors$ to refer to $\{e \in E \mid \ell(e) = c\}$):
	\begin{alignat*}{2}
		D[v,c,S] \coloneqq \max_{\substack{S_2 \subseteq C_2 (v) \\ S_3 \subseteq C_3(v) \cap E_c}}& 	\abs{S_2} +  \abs{S_3}\\ 
		 & + \sum_{w \in \ch_c(v)} \max &&\{ 1 + D[w,c,(S \cup S_2 \cup S_3) \cap (C_1(w) \cup C_4(w))],\\
		 & && \max_{c' \in \colors\setminus \{c\}} D[w,c',(S \cup S_2 \cup S_3) \cap (C_1(w) \cup C_4(w))]\} \\
		 & + \sum_{w \in \ch(v) \setminus \ch_c(v)} && \max_{c' \in \colors}  D[w,c',(S \cup S_2 \cup S_3) \cap (C_1(w) \cup C_4(w))].
	\end{alignat*}

	Having computed all table entries, the result can be found in $D[r,\emptyset]$ (note that $C_1(r) \cup C_4(r) = \emptyset$).
	
	For each $v\in V$, the table $D$ has $\abs{\colors}2^{\abs{C_1(v)} + \abs{C_4(v)}}$ entries.
	Computing each entry requires at most $2\cdot\Delta(T)\cdot2^{\abs{C_2(v)} + \abs{C_3(v)}} \cdot \abs{C}$ look-ups where $\Delta(T)$ is the maximum degree in $T$.
	From this, we get a worst-case running time of $\bigO(2^{\lfe(G,T)}\cdot \abs{C}^2\cdot n^2)$.
	The correctness of this algorithm follows directly by induction on $T$.
\end{proof}

By contrast to local feedback edge number, we will now show that \cec{} is \W{1}-hard with respect to tree-cut width, which is a lower bound for local feedback edge number, and for vertex cover number, which is incomparable to local feedback edge number.
These two hardness results, in particular the one for vertex cover number, rule out FPT algorithms for many other structural graph parameters, including treewidth, treedepth, feedback vertex number, etc.

Tree-cut width is a parameter introduced fairly recently by Wollan~\cite{Wollan2015}.
The algorithmic uses of this parameter were systematically investigated by Ganian et al.~\cite{Ganian2015}.
It is defined as follows.
A \emph{rooted tree-cut decomposition} of~$G=(V,E)$ is a triple~$(T,\calX,t_0)$ where $T=(W,F)$ is a tree and $\calX = \{X_t \subseteq V \mid t \in W\}$ where the elements of $\calX$ are pairwise disjoint and $\bigcup_{t \in W} X_t = V$.
The node $t_0\in W$ is the root of~$T$.
For any edge $e=\{t_1,t_2\}\in F$, $T-e$ has two connected components $T_1$ and $T_2$.
Let~$\mathrm{cut}(e) \coloneqq \{ \{u,v\} \in E \mid u\in \bigcup_{t \in T_1}X_t \text{ and } v \in \bigcup_{t \in T_2}X_t\}$.
For $t\in W\setminus\{t_0\}$, let $e(t)$ be the edge incident to $t$ which is on the unique path from $t$ to~$t_0$.
The \emph{adhesion} of~$t\in W\setminus\{t_0\}$ is~$\mathrm{adh}(t) \coloneqq \abs{\mathrm{cut}(e(t))}$, while that of the root is $\mathrm{adh}(t_0)\coloneqq 0$.
For any node~$t\in W$, let~$T_1,\ldots,T_s$ be the connected components of $T-t$ and $Z_1,\ldots,Z_t \subseteq V$ with $Z_i \coloneqq \bigcup_{b \in T_i} X_b$ the vertices in bags of each component.
The \emph{torso} $H_t$ of $t$ is the graph obtained from $G$ by consolidating each $Z_i$ into a single vertex, that is by replacing $Z_i$ with a single vertex $v_i$ and adding an edge between $v_i$ and any neighbor of a member of $Z_i$, thereby possible creating parallel edges.
\emph{Suppressing} a vertex of degree~$1$ means deleting it, while \emph{suppressing} a vertex of degree~$2$ means replacing it with an edge between its two neighbors.
The \emph{torso size} of $t$ is $\mathrm{tor}(t)$ is the number of vertices in the graph obtained from the torso $H_t$ by exhaustively suppressing vertices in $H_t - X_t$ of degree at most two.
The \emph{width} of $(T,\calX,t_0)$ is $\max \{\mathrm{adh}(t),\mathrm{tor}(t) \mid t \in W\}$.
The tree-cut width of $G$ is the minimum width of any tree-cut decomposition of $G$.

\begin{theorem}
	\label{thm:w1-vc}
	\cec{} on bipartite graphs is \W{1}-hard with respect to both vertex cover number and tree-cut width.
\end{theorem}
\begin{proof}
	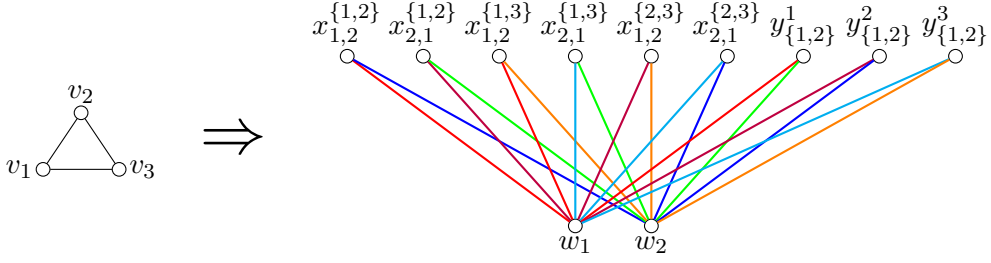
\begin{figure}[t]
		\centering
		\begin{tikzpicture}[yscale=0.75]
		\tikzpreamble
		\node[xnode] (v1) at (0,0) {};
		\node[left] at (v1) {$v_1$};
		\node[xnode] (v2) at (0.5,1) {};
		\node[above] at (v2) {$v_2$};
		\node[xnode] (v3) at (1,0) {};
		\node[right] at (v3) {$v_3$};
		\draw (v1) -- (v2) -- (v3) -- (v1);
		
		\node () at (2.5,0.5) {\Huge $\Rightarrow$};
		
		\begin{scope}[rotate=90,yshift=-8cm,xshift=-5cm]
		\node[xnode] (w1) at (4,1) {};
		\node[below] at (w1) {$w_1$};
		\node[xnode] (w2) at (4,0) {};
		\node[below] at (w2) {$w_2$};

		\node[xnode] (x1212) at (7,4) {};
		\node[above] at (x1212) {$x^{\{1,2\}}_{1,2}$};
		\draw[red,thick] (w1) -- (x1212);
		\draw[blue,thick] (w2) -- (x1212);
		
		\node[xnode] (x1221) at (7,3) {};
		\node[above] at (x1221) {$x^{\{1,2\}}_{2,1}$};
		\draw[purple,thick] (w1) -- (x1221);
		\draw[green,thick] (w2) -- (x1221);
		
		\node[xnode] (x1312) at (7,2) {};
		\node[above] at (x1312) {$x^{\{1,3\}}_{1,2}$};
		\draw[red,thick] (w1) -- (x1312);
		\draw[orange,thick] (w2) -- (x1312);
		
		\node[xnode] (x1321) at (7,1) {};
		\node[above] at (x1321) {$x^{\{1,3\}}_{2,1}$};
		\draw[cyan,thick] (w1) -- (x1321);
		\draw[green,thick] (w2) -- (x1321);
		
		\node[xnode] (x2312) at (7,0) {};
		\node[above] at (x2312) {$x^{\{2,3\}}_{1,2}$};
		\draw[purple,thick] (w1) -- (x2312);
		\draw[orange,thick] (w2) -- (x2312);
		
		\node[xnode] (x2321) at (7,-1) {};
		\node[above] at (x2321) {$x^{\{2,3\}}_{2,1}$};
		\draw[cyan,thick] (w1) -- (x2321);
		\draw[blue,thick] (w2) -- (x2321);
		
		\node[xnode] (y112) at (7,-2) {};
		\node[above] at (y112) {$y^{1}_{\{1,2\}}$};
		\draw[red,thick] (w1) -- (y112);
		\draw[green,thick] (w2) -- (y112);
		
		\node[xnode] (y212) at (7,-3) {};
		\node[above] at (y212) {$y^{2}_{\{1,2\}}$};
		\draw[purple,thick] (w1) -- (y212);
		\draw[blue,thick] (w2) -- (y212);
		
		\node[xnode] (y312) at (7,-4) {};
		\node[above] at (y312) {$y^{3}_{\{1,2\}}$};
		\draw[cyan,thick] (w1) -- (y312);
		\draw[orange,thick] (w2) -- (y312);
		\end{scope}
		\end{tikzpicture}
		\caption{An illustration of the reduction used to prove \cref{thm:w1-vc}.
			The edge colors in the figure represent the colors defined in the reduction as follows:
			\textbf{\color{red} Red}:~$c^1_1$, \textbf{\color{purple} Purple}:~$c^2_1$,
			\textbf{\color{cyan} Cyan}:~$c^3_1$,   \textbf{\color{green} Green}:~$c^1_2$, \textbf{\color{blue} Blue}:~$c^2_2$, 
			\textbf{\color{orange} Orange}:~$c^3_2$.}
		\label{fig:w1-vc}
	\end{figure}
	We will give a parameterized reduction from \textsc{Independent Set} on $d$-regular graphs.
	This problem is \W{1}-hard when parameterized by the size~$s$ of the sought independent set~\cite{Mathieson2012}.
	The reduction is illustrated in \Cref{fig:w1-vc}.
	Let $(G=(V,E),s)$ be an instance of \textsc{Independent Set} where $G$ is a $d$-regular graph and assume that $V=\{v_1,\ldots,v_n\}$.
	
	We will construct an instance $(G'=(V',E'),\ell,k\coloneqq s(s-1)(d+1))$ of \cec{}.
	We choose the set of colors as $\colors \coloneqq \{ c^a_i \mid a\in[n], i\in [s]  \}$.
	The idea is that the color $c^a_i$ represents~$v_a$ being chosen as the $i$-th vertex in the independent set.
	We construct~$G'$ as follows:
	We start with vertices $w_1,\ldots,w_s$.
	For each pair $(i,j) \in [s]\times [s]$ and each edge~$\{v_a,v_b\} \in E$ we add to $G'$ a vertex~$x_{i,j}^{\{a,b\}}$, which we connect to $w_i$ by an edge with color~$c_i^a$ and to $w_j$ by an edge with color $c_j^b$.
	For each pair $\{i,j\} \in \binom{[s]}{2}$ and each vertex $v_a \in V$, we also add to $G'$ a vertex $y^a_{\{i,j\}}$, which we connect to $w_i$ by an edge with color~$c_i^a$ to $w_j$ by an edge with color $c_j^a$.
	
	Clearly, this instance can be computed in polynomial time.
	It remains to show that $(G,s)$ is a \yes-instance for \textsc{Independent Set} if and only if $(G',r, k,\ell)$ is a \yes-instance for \cec{}.
	
	\RD{}
	Let $\{v_{a_1},\ldots,v_{a_s}\}$ be an independent set in $G$.
	Let $\bot \in C$ denote an arbitrary default color.
	We let $f(w_i) \coloneqq c^{a_i}_i$, 
	\begin{align*}
	f(x_{i,j}^{\{a,b\}})\coloneqq
	\begin{cases}
	c_i^{a}, & \text{ if } a = a_i, \\
	c_j^{b}, & \text{ if } b = a_j, \\
	\bot, & \text{ otherwise,}
	\end{cases}
	\quad \text{ and }
	f(y^a_{\{i,j\}}) \coloneqq
	\begin{cases}
	c^a_{i}, & \text{ if } a = a_i, \\
	c^a_{j}, & \text{ if } a = a_j, \\
	\bot, & \text{ otherwise.}
	\end{cases}
	\end{align*}
	We note that $f(x_{i,j}^{\{a,b\}})$ is well-defined in the sense that the first two cases are disjoint, because~$v_a$ and $v_b$ cannot both be in the independent set if they are adjacent.
	We claim that there are~$k=s(s-1)(d+1)$ stable edges under $f$.
	We prove this by showing that each $w_i$ is incident to $(s-1)(d+1)$ stable edges.
	Let $v_{b_1},\ldots,v_{b_d}$ be the neighbors of $v_{a_i}$ in $G$.
	Then, for each $j \in [s]\setminus \{i\}$, the edges between~$w_i$ and $x_{i,j}^{\{a_i,b_1\}},\ldots,x_{i,j}^{\{a_i,b_d\}}$ and $y^{a_i}_{i,j}$ are stable.
	
	\LD{} Suppose that $F$ is a stable edge set in $G'$ of size at least $k$.
	Since each vertex $w_i$ is incident to just $(s-1)(d+1)$ edges in the same color, each $w_i$ can be incident to no more than $(s-1)(d+1)$ stable edges.
	Moreover, since every edge in $G'$ is incident to a $w_i$, it follows that each $w_i$ must be incident to exactly $(s-1)(d+1)$ stable edges.
	For each $i\in [s]$, choose $a_i$ such that $f_F(w_i) = c^{a_i}_i$.
	We claim that $\{v_{a_1},\ldots,v_{a_s}\}$ is an independent set in $G$.
	If $G$ contains an edge $\{v_{a_i},v_{a_j}\}$, then consider the color~$f_F(x^{\{a_i,b_j\}}_{i,j})$.
	\Wilog{}, we may assume that this color is $c^a_i$.
	Then, the edge $\{w_j,x^{\{a_i,a_j\}}_{i,j}\}$, which has color $c^{a_j}_j = f_F(w_j)$, is unstable.
	It follows that $w_j$ cannot be incident to $(s-1)(d+1)$ stable edges.
	If $v_{a_1},\ldots,v_{a_s}$ are not pairwise distinct, a similar argument using a vertex $y^a_{\{i,j\}}$ yields a similar contradiction.
	It follows that $\{v_{a_1},\ldots,v_{a_s}\}$ is, in fact, an independent set of size $s$.
	
	It remains to show that the vertex cover number and the tree-cut width of $G'$ are bounded in $s$.
	Regarding the vertex cover number, it is easy to see that $\{w_1,\ldots,w_s\}$ is a vertex cover in $G'$.
	Next, we will give a tree-cut decomposition of $G'$.
	We place $w_1,\ldots,w_s$ in one bag, the root of $T$, and every other vertex in its own bag.
	All singleton bags are direct children of the root.
	By definition, the adhesion of the root is $0$, while the adhesion of each singleton bag is $2$, since all vertices other than $w_1,\ldots,w_s$ have degree~$2$.
	The singleton bags have torso size at most~$2$, while the size of the torso of the root is at most $s$, since all vertices except $w_1,\ldots,w_s$ have degree~$2$ and are therefore suppressed in the torso of the root.
	It follows that the width of this tree-cut decomposition is at most~$s$.
\end{proof}

Cai and Leung~\cite{Cai2018} showed that \cec{} is \NP-hard, even if the maximum degree in the input graph is at most four, the input graph is planar and bipartite, and there are only three colors.
We conclude by strengthening this result and showing that \cec{} is also \NP-hard on cubic graphs, that is graphs with maximum degree three.

\begin{theorem}
	\label{thm:nphard-cubic}
	\cec{} is NP-hard even if every vertex has degree at most three and $|C| = 5$.
\end{theorem}
\begin{proof}
	We reduce from \textsc{3-SAT}.
	The input consists of a set of variables $X$ and a set of clauses~$Y$, each consisting of at most three literals and the task is to find a truth assignment satisfying every clause.
	We may assume that every variable $x \in X$ appears at most three times in $Y$ \cite{Tovey84}.
	We may also assume that for each variable $x \in Y$, the literals $x$ and $\neg x$ each appears at least once in $Y$, since otherwise we can simplify the instance by deleting all clauses containing the variable $x$.

	Let $(X, Y)$ be an instance of \textsc{3-SAT} with $n = |X|$ variables and $m = |Y|$ clauses.
	We construct an instance $(G = (V, E), \ell, k)$ of \cec{} with $|C| = 5$ as follows.
	The graph $G$ contains a vertex $u_x$ for every variable $x \in X$ and a vertex $v_y$ for every clause~$y \in Y$.
	There is an edge $\{ u_x, v_y \}$ if $x$ or $\neg x$ appears in $y$.
	This concludes the construction of~$G$.
	For every edge $e \in E$, we define its color as follows.
	We assume that $X = \{ x_1, \dots, x_n \}$ and we color the edges incident to $v_{x_i}$ for increasing $i$.
	Suppose that $x_i$ appears in three clauses $y_i^1, y_i^2, y_i^3$.
	\Wilog{}, we assume that $x_i$ appears positively (negatively) in $y_i^1$ and $y_i^2$ and negatively (positively) in $y_i^3$.
	We will color $\{ u_{x_i}, v_{y_i^1} \}$ and $\{ u_{x_i}, v_{y_i^2} \}$ in the same color  $c^{12}$ and $\{ u_{x_i}, v_{y_i^3} \}$ in a different color $c^3$.
	We choose $c^{12}$ and $c^3$ in such a way that $v_{y_i^1}$ and $v_{y_i^2}$ have no other incident edge colored in $c^{12}$ and $v_{y_i^3}$ has no other incident edge colored in $c^3$.
	Since $v_{y_i^j}$ for each $j \in [3]$ has degree at most three and $|C| \ge 5$, we can always choose $c^{12}$ and $c^3$ this way.
	Finally, let $k \coloneqq m = |Y|$.
	Essentially, the variable $x$ will be true (or false if there are two negative occurrences of $x$) if the vertex $v_x$ is colored in $c^{12}$. 

	\RD{}
	Suppose that there is a satisfying assignment $\varphi$.
	Then, for every clause $y$, let $L_y$ be an arbitrary literal in $y$ satisfied by $\varphi$ and let $x_y$ be its variable, that is, $L_y = x_y$ or $L_y = \neg x_y$. 
	We claim that the edge set $\{ \{ u_{x_y}, v_y \} \mid y \in Y \}$ is stable.
	To that end, consider a coloring $f \colon V \to C$ such that $f(w)$ is  the color of an edge in $F$ incident to $w$.
	Note that $f$ is well-defined:
	For every $x \in X$, there may more than edge in $F$ incident to $u_x$ but they have the same color.
	Moreover, for every $y \in Y$, there is exactly one edge in $F$ incident to $v_y$.
	Thus, $F$ is stable under $f$.

	\LD{}
	Conversely, suppose that there is a stable set $F \subseteq E$ of size $k = \abs{Y}$.
	Note that for every $y \in Y$, $v_y$ is incident to edges of distinct color.
	Thus, there is at most one edge of $F$ incident to $v_y$.
	On the other hand, every edge in $E$ has exactly one endpoint in $\{ v_y \mid y \in Y \}$.
	It follows that every vertex $v_y$ is incident to exactly one edge of $F$.
	Now consider a truth assignment $\varphi$ where $x$ is true (false) if $\{ u_x, v_y \} \in F$ and $x$ ($\neg x$, respectively) is in the clause $y$ for each variable $x \in X$.
	Since $F$ is stable, this truth assignment is well-defined. 
	Moreover, it is a satisfying assignment.
\end{proof}

For hypergraphs, one can show \NP-hardness in an even more restricted setting:

\begin{theorem}
	\label{thm:max-degree-2}
	\cec{} is \NP-hard on hypergraphs with maximum degree~$\Delta=2$, order~$d=3$, and $k=3$ colors.
\end{theorem}
\begin{proof}
	We will give a polynomial-time many-to-one reduction from the \NP-complete~\cite{Garey1977} \prob{Independent Set} with maximum degree~$3$.
	
	Let $(G=(V,E),s)$ be an instance of \prob{Independent Set} where $G$ is a graph with maximum degree~$3$ and $s$ is a nonnegative integer.
	We subdivide every edge twice and increase $s$ by $\abs{E}$.
	This is correct, because subdividing an edge twice increases the independence number of a graph by exactly~$1$.	
	We compute a proper $3$-coloring $f \colon V \to [3]$ of $G$.	
	Such a coloring exists and can be computed in polynomial time by Brooks's theorem~\cite{Brooks41,Lovasz1975}, which states that every graph with maximum degree~$\Delta\geq 3$ that does not contain a clique of size $\Delta+1$ is $\Delta$-colorable.
	Because every edge was subdivided twice, $G$ cannot contain a clique of size~$4$ and, hence, is $3$-colorable.
	We output an instance $(G'=(V',E'),\ell,k)$ for \cec{}.
	We define
	\begin{align*}
	V' &\coloneqq E, \\
	e_v & \coloneqq \{e \in E \mid v \in e\} \text{ for each } v \in V,\\
	E' &\coloneqq \{e_v \mid v \in V\}, \text{ and }\\
	\ell (e_v) & \coloneqq f(v) \text{ for each } v \in V.
	\end{align*}
	Finally, we let $k \coloneqq s$.
	
	Clearly, this reduction can be computed in polynomial time.
	For each $v\in V$, the degree of $v$ is at most $3$ and, therefore, $\abs{e_v} \leq 3$.
	Hence, the order of $G'$ is at most $3$.
	Each edge $e \in E$ contains exactly two vertices.
	Hence, the degree of $e$ in $G'$ is at most $2$.
	It remains to show that $G$ contains an independent set of size at most $s$ if and only if $G'$ contains a stable edge set of size $k$.
	
	\RD{} Let $X \subseteq V$ with $\abs{X} \geq s$ be an independent set in $G$.
	Then, $\{e_v \mid v \in X\}$ is a stable edge set of size at least $k$.
	
	\LD{} Suppose that $F \subseteq E'$ with $\abs{F} \geq k$ is a stable edge set in $G'$.
	Since $\ell(e_v) \neq \ell(e_{v'})$ if $e_v \cap e_{v'} \neq \emptyset$, it follows that the edges in $F$ are pairwise disjoint.
	Hence, there is no edge~$\{v,v'\} \in E$ such that $e_v,e_{v}' \in F$.
	Therefore, $X \coloneqq \{v \in V \mid e_v \in F\}$ is an independent set of size at least $s$ in $G$.
\end{proof}

\textsc{Colored Hypergraph Clustering} is polynomial-time solvable if the maximum degree is at most $1$ or if $\abs{C} \leq 2$.
In the first case, the conflict graph is empty and, in the second case, it is bipartite.

\section{Conclusion}
Our results in many ways complete and extend the picture of the parameterized complexity of \cec{} initiated by \citet{Cai2018}.
For the parameterization by the number of stable edges, we have given an improved algorithm and a polynomial kernel.
We have also initiated the study of strictly smaller parameters than both the number of stable edges $k$ and the number of unstable edges $r$.
Finally, we gave a picture of the problem's parameterized complexity for structural graph parameters.
We conclude by listing a few open problems and avenues for further research:
\begin{itemize}
	\item Can the kernel (\cref{thm:fpt-stable-edges}) be improved to size~$\bigO(k^2)$?
	\item Are there other natural (tight) lower bounds for $r$ or $k$, besides the degree-based and matching based lower bounds we considered, that yield fixed-parameter algorithms?
	\item One issue with the structural parameters we have considered is that they are oblivious to the complexity introduced by the edge colors.
	It may be useful to consider structural parameters that explicitly take the structure of the edge coloring into account.
	Such parameters were studied by Morawietz et al.~\cite{Morawietz2020}.
	Unfortunately, most of those parameters are smaller than the number of colors and, therefore, of little use in the context of \cec{}.
	However, \cref{prop:distance-to-poly} implies that \cecS{} is \FPT{} with respect to $m_{>2}$ (in the terminology used by Morawietz et al.~\cite{Morawietz2020}), and the NP-hardness of \cecS{} on tricolored graphs implies para-NP-hardness with respect to $m_{>3}$.
	It may be an interesting challenge to develop color-sensitive structural parameters that lead to useful \FPT{} algorithms for \cec{}.
\end{itemize}

\section*{Acknowledgments}
The second author is supported by the DFG Project DiPa, NI 369/21.
The third author is supported by the DFG Research Training Group 2434 ``Facets of Complexity''.

We thank Till Fluschnik (TU Clausthal) and Klaus Heeger (TU Berlin) for a fruitful discussion which led to the NP-hardness for cubic graphs (\cref{thm:nphard-cubic}).

This paper is dedicated to Rolf, our co-author, colleague, and advisor.
Rolf's tremendous contributions to computer science, particularly to parameterized algorithmics, will be dearly missed.
The computer science community will build on the foundations he has laid.

{\small
	\bibliography{strings-long,cec}
}
\end{document}